\documentclass[nonacm]{acmart}

\copyrightyear{2025}
\acmYear{2025}
\setcopyright{cc}
\setcctype{by}
\acmConference[KDD '25]{Proceedings of the 31st ACM SIGKDD Conference on Knowledge Discovery and Data Mining V.2}{August 3--7, 2025}{Toronto, ON, Canada}
\acmBooktitle{Proceedings of the 31st ACM SIGKDD Conference on Knowledge Discovery and Data Mining V.2 (KDD '25), August 3--7, 2025, Toronto, ON, Canada}
\acmDOI{10.1145/3711896.3736940}
\acmISBN{979-8-4007-1454-2/2025/08}

\renewcommand\footnotetextcopyrightpermission[1]{} 
\usepackage{times}  
\usepackage{helvet}  
\usepackage{courier}  
\usepackage{graphicx} 
\usepackage{xcolor,ifthen}
\usepackage{multirow}
\usepackage{balance}
\usepackage{amsthm, thmtools}
\usepackage{amsfonts}
\usepackage{paralist}
\usepackage{nameref}

\usepackage{amsmath}
\usepackage{natbib}  
\usepackage{caption} 

\def\notationmismatch{true} 
\ifdefined\notationmismatch

\else

\fi

\usepackage{url}
\usepackage{amsfonts}
\usepackage{tikz}
\usetikzlibrary{decorations.text, arrows.meta}
\usepackage{pgfplots}
\pgfplotsset{compat=newest}
\usepgfplotslibrary{fillbetween}
\usepgfplotslibrary{statistics}
\usetikzlibrary{pgfplots.statistics, pgfplots.colorbrewer} 
\usepackage{pgfplotstable}

\usepackage{bm,mathtools}
\usepackage{algorithmic}
\usepackage{subcaption}
\definecolor{ForestGreen}{rgb}{0.1333,0.5451,0.1333}
\definecolor{DarkRed}{rgb}{0.8,0,0}
\definecolor{Red}{rgb}{1,0,0}
\usepackage{cleveref}
\usepackage{xspace}
\usepackage[normalem]{ulem}
\usepackage{comment}
\usetikzlibrary{intersections}
\usepgfplotslibrary{fillbetween}
\usepackage{cancel}

\newcommand{\brac}[1]{\left[ #1 \right]}
\newcommand{\curly}[1]{\left\{ #1 \right\}}
\newcommand{\paren}[1]{\left( #1 \right)}

\usepackage{footnote}
\makesavenoteenv{tabular}
\makesavenoteenv{table}

\def\E{\mathbb{E}}

\def\ShowComment{true} 

\ifdefined\ShowComment

\def\antonio#1{\marginpar{$\leftarrow$\fbox{af}}\footnote{$\Rightarrow$~{\sffamily #1 --Antonio}}}
\def\jaya#1{\marginpar{$\leftarrow$\fbox{jp}}\footnote{$\Rightarrow$~{\sffamily #1 --Jaya}}}
\def\jose#1{\marginpar{$\leftarrow$\fbox{ja}}\footnote{$\Rightarrow$~{\sffamily #1 --Jose}}}
\def\juan#1{\marginpar{$\leftarrow$\fbox{jr}}\footnote{$\Rightarrow$~{\sffamily #1 --Juan}}}
\def\sergio#1{\marginpar{$\leftarrow$\fbox{sd}}\footnote{$\Rightarrow$~{\sffamily #1 --Sergio}}}
\def\mohit#1{\marginpar{$\leftarrow$\fbox{md}}\footnote{$\Rightarrow$~{\sffamily #1 --Mohit}}}
\def\rosa#1{\marginpar{$\leftarrow$\fbox{rl}}\footnote{$\Rightarrow$~{\sffamily #1 --Rosa}}}

\else

\def\antonio#1{}
\def\jaya#1{}
\def\jose#1{}
\def\juan#1{}
\def\sergio#1{}
\def\mohit#1{}
\def\rosa#1{}

\fi

\makeatletter
\def\thmt@refnamewithcomma #1#2#3,#4,#5\@nil{%
	\@xa\def\csname\thmt@envname #1utorefname\endcsname{#3}%
	\ifcsname #2refname\endcsname
	\csname #2refname\expandafter\endcsname\expandafter{\thmt@envname}{#3}{#4}%
	\fi
}
\makeatother

\declaretheorem[refname={Theorem,Theorems},Refname={Theorem,Theorems},name={Theorem}]{theorem}

\declaretheorem[refname={Lemma,Lemmas},Refname={Lemma,Lemmas},name={Lemma}]{lemma}

\declaretheorem[refname={Corollary,Corollaries},Refname={Corollary,Corollaries},name={Corollary}]{corollary}

\declaretheorem[refname={Assumption,Assumptions},Refname={Assumption,Assumptions}]{assumption}

\declaretheorem[refname={Definition,Definitions},Refname={Definition,Definitions},name={Definition}]{defn}

\declaretheorem[numberlike=theorem,refname={Claim,Claims},Refname={Claim,Claims}]{claim}

\def\Paren#1{\left( #1 \right)}		
\def\Brack#1{\left[ #1 \right]}	
\def\E{\mathbb{E}}			

\def\cases#1{C_{#1}}
\def\reach#1{R_{#1}}
\def\Cases#1{C_{#1}}
\def\Reach#1{R_{#1}}

\def\method{\mathcal{M}}
\def\MoR{\ensuremath{\mathrm{MoR}}\xspace}
\def\RoS{\ensuremath{\mathrm{RoS}}\xspace}

\def\pdeg#1{P_\text{deg}\Paren{#1}}

\newcommand{\prevalence}{{{prevalence}}\xspace}
\newcommand{\ratio}[1][]{\ensuremath{\rho #1}\xspace}
\newcommand{\estimate}[2]{\ensuremath{\hat{\rho}_{#1} #2}\xspace}
\newcommand{\error}[3]{\ensuremath{\mathcal{E}_{#1}^{#3} #2}\xspace}

\usepackage{enumitem}
\setlist[itemize]{noitemsep, topsep=0pt, leftmargin=3.5mm}

\usepackage{thm-restate}
\newcommand{\proofinappendix}{\em (Proof in \Cref{sec:omitted-proofs}.)}

\newtheorem{definition}{Definition}

\newif\ifannote
\annotetrue  

\ifannote
    
    \newcommand{\anncomment}[3]{{\color{#1}[#2: #3]}}
\else
    
    \newcommand{\anncomment}[3]{}
\fi

\author{Sergio Díaz-Aranda}
\orcid{0000-0002-6919-9236}
\affiliation{%
  \institution{IMDEA Networks Institute\\
  Universidad Carlos III de Madrid}
  \state{Madrid}
  \country{Spain}
}
\email{sergio.diaz@imdea.org}

\author{Juan Marcos Ramirez}
\orcid{0000-0003-0000-1073}
\affiliation{%
  \institution{IMDEA Networks Institute}
  \streetaddress{Av. Mar Mediterr\'aneo 22}
  \city{Legan\'es}
  \state{Madrid}
  \country{Spain}
}
\email{juan.ramirez@imdea.org}

\author{Mohit Daga}
\orcid{0000-0002-7579-156X}
\affiliation{%
  \institution{KTH Royal Institute of Technology}
  \department{Division of Theoretical Computer Science, EECS, KTH}
  \city{Stockholm}
  \country{Sweden}
}
\email{mdaga@kth.se}

\author{Jaya Prakash Champati}
\orcid{0000-0002-5127-8497}
\affiliation{%
  \institution{University of Victoria}
  \city{Victoria}
  \state{BC}
  \country{Canada}
}
\email{jpchampati@uvic.ca}

\author{Jose Aguilar}
\orcid{0000-0003-4194-6882}
\affiliation{%
  \institution{IMDEA Networks Institute}
  \streetaddress{Av. Mar Mediterr\'aneo 22}
  \city{Legan\'es}
  \state{Madrid}
  \country{Spain}
}
\affiliation{%
  \institution{Universidad de Los Andes}
  \city{M\'erida}
  \country{Venezuela}
}
\email{aguilar@ula.ve}

\author{Rosa Lillo}
\orcid{0000-0003-0802-4691}
\affiliation{%
  \institution{Universidad Carlos III de Madrid}
  \department{uc3m-Santander Big Data Institute}
  \city{Getafe}
  \state{Madrid}
  \country{Spain}
}
\email{rosaelvira.lillo@uc3m.es}

\author{Antonio {Fern\'andez Anta}}
\orcid{0000-0001-6501-2377}
\affiliation{%
  \institution{IMDEA Software Institute\\
  IMDEA Networks Institute}
  \state{Madrid}
  \country{Spain}
}
\email{antonio.fernandez@imdea.org}

\renewcommand{\shortauthors}{Sergio Díaz-Aranda et al.}

\title{Error Bounds for the Network Scale-Up Method}

\date{}

\begin{document}

\begin{abstract}
Epidemiologists and social scientists have used the Network Scale-Up Method (NSUM) for over thirty years to estimate the size of a hidden sub-population within a social network. This method involves querying a subset of network nodes about the number of their neighbors belonging to the hidden sub-population. In general, NSUM assumes that the social network topology and the hidden sub-population distribution are well-behaved; hence, the NSUM estimate is close to the actual value. However, bounds on NSUM estimation errors have not been analytically proven. This paper provides analytical bounds on the error incurred by the two most popular NSUM estimators. These bounds assume that the queried nodes accurately provide their degree and the number of neighbors belonging to the hidden sub-population. Our key findings are twofold. First, we show that when an adversary designs the network and places the hidden sub-population, then the estimate can be a factor of $\Omega(\sqrt{n})$ off from the real value (in a network with $n$ nodes). Second, we also prove error bounds when the underlying network is randomly generated, showing that a small constant factor can be achieved with high probability using samples of logarithmic size $O(\log n)$. We present improved analytical bounds for Erdős–Rényi and Scale-Free networks. Our theoretical analysis is supported by an extensive set of numerical experiments designed to determine the effect of the sample size on the accuracy of the estimates in both synthetic and real networks.
\end{abstract}

\keywords{Aggregated Relational Data; Network Scale-Up Method; Ratio of Sums; Mean of Ratios; Error Bounds; Hidden Population}



\maketitle

\newcommand\kddavailabilityurl{https://doi.org/10.5281/zenodo.15575415}

\ifdefempty{\kddavailabilityurl}{}{
\begingroup\small\noindent\raggedright\textbf{KDD Availability Link:}\\
The source code of this paper has been made publicly available at \url{\kddavailabilityurl}.
\endgroup
}

\section{Introduction}

In a society, estimating the size of hidden or hard-to-reach sub-populations (such as individuals affected by disaster, disease, or clandestine networks) is critical for developing targeted strategies to address the challenges or threats associated with these hidden groups. Direct enumeration methods for estimating the hidden sub-population size are often impractical due to these groups' dispersed, unreachable, or secretive nature. The Network Scale-Up Method (NSUM) \cite{bernard1988many,bernard1991estimating} is an innovative solution that addresses this challenge by leveraging the social network to estimate the size.
Instead of directly querying individuals whether they belong to the hidden population, NSUM uses \textit{indirect reporting}.
NSUM collects data through queries of the type ``How many neighbors do you know?'' and ``How many of those neighbors belong to the hidden sub-population?" The answers to these questions form the Aggregated Relational Data (ARD), which is used to estimate the size of the hidden sub-population. Using indirect reporting increases the privacy and the proportion of the population a survey reaches, which, combined with the simplicity of the NSUM estimator, makes it very powerful.
Thus, it has been widely used in applications including estimating the number of affected people during disasters \cite{bernard1988many, bernard2001estimating}, social networks audience \cite{DBLP:conf/chi/BernsteinBBK13}, security profiling \cite{crawford2016graphical}, and epidemic prevalence and evolution (e.g., AIDS \cite{unaids2010guidelines}, COVID-19 \cite{garcia2021estimating}). 

There is a suite of NSUM methods that have been used in the literature \cite{zheng2006many,mccormick2010many,mccormick2012latent,maltiel2015estimating}, and they differ in the way they use the ARD to provide the estimates \cite{laga2021thirty-2021-jasa}. The most popular NSUM estimators are the Mean of Ratios (\MoR) \cite{killworth1998social} and Ratio of Means (\RoS) \cite{killworth1998estimation} estimators. While \MoR was the first estimator proposed for \prevalence \cite{killworth1998social,laga2021thirty-2021-jasa}, relying on averaging the naive estimates of each individual, \RoS became the more widely used estimator in the NSUM literature \cite{ahmadi2019twelve,ezoe2012population,killworth1998estimation,snidero2012scale,wang2015application}.

Although NSUM has been widely used, to the best of our knowledge, no existing work provides analytical guarantees on the quality of the aforementioned NSUM estimators. The main challenge in analyzing the NSUM estimate is due to the indirect reporting, which may result in multiplicities in the reported hidden nodes,
thereby introducing correlations in the ARD. Recently, the authors in \cite{neurips-chen-2016,srivastava2024nowcasting} have provided analytical bounds from indirect reporting.
Chen et al.  \cite{neurips-chen-2016} estimated the size of the network and the hidden communities, providing analytical guarantees for Erdős-Rényi and Stochastic Block Model topologies. However, they assume that the number of hidden nodes in the sample is known, which is not the case in ARD.
Srivastava et al.~\cite{srivastava2024nowcasting} estimated the trend (not the value) of the hidden sub-population \prevalence over time. They show that, if the degree of the social network has a constant mean over time and its variance is small,
the trend of the hidden sub-population \prevalence can be estimated using the ARD, and the error is smaller than with direct reporting.


\noindent
\textbf{Our Contributions}
This paper partially fills the void of NSUM analytical bounds by providing a theoretical performance analysis for NSUM estimators in general networks.
Let $\ratio{}$ denote the fraction of the population that belongs to the hidden sub-population and $\estimate{}$ is its NSUM estimate.
The \textit{error} of the estimator $\estimate{}$ is how far it is from $\ratio{}$ (as the ratio between them, or equivalently, the factor that transforms one into the other).
These are the main results.
\begin{itemize}
    \item A surprising result we prove is that, in adversarial settings, any NSUM estimator has an error that grows with $\sqrt{n}$, where $n$ is the size of the network, even when the ARD is collected from all the nodes in the network. 
    \item For general random networks with a fixed hidden population, we provide analytically probabilistic upper bounds for $\MoR$ and $\RoS$ estimators. We address the challenge of the duplicates by proving \textit{negative correlation} between the underlying random variables. The results upper bound the probability that the error is above a fixed threshold $(1+\epsilon)$. Then, we use these bounds to show that, with both estimators, it is enough to collect ARD from a logarithmic number of nodes in a random network to have a small error of $(1+\epsilon)$ with high probability. We also show that $\RoS$ has smaller error bounds than $\MoR$.
    \item As special cases of these general random networks, we study two popular classes of random networks, Scale-Free \cite{clauset2009power} and Erdős–Rényi networks \cite{erdosrenyi}, and compute probabilistic error bounds for \MoR and \RoS estimates. 
    \item
    We evaluate the goodness of the bounds using simulation. We observe that the analytical bounds are close to the actual empirical measures. We also observe that, as hinted by the analytical bounds, the estimates obtained with $\RoS$ have lower errors than those obtained with $\MoR$, which is especially evident in the Scale-Free networks.
    \item
    Finally, we apply the bounds to real networks. We observe that the bounds are still close to the actual empirical measures, even though the networks may not exactly be the type of random networks assumed in the analysis.
\end{itemize}

\section{Model and Definitions\label{s-contribs}}


\noindent
\textbf{Network model.} 
We model a social network as a directed graph $G=(V,E)$\footnote{We will indistinctly use graph and network, vertex and node, and edge and link.},
in which $V$ is the population under study.
A subset $H \subseteq V$ is the \emph{hidden sub-population} of the network, which is the set of network members $G$ with a given property. 
A directed edge $(u,v) \in E$ implies that node $v$ knows node $u$ (in-neighbor) and whether it belongs to the hidden sub-population $H$. 
A \textit{bidirectional network} is a directed graph, where all edges are bidirectional, i.e., $\forall u,v \in V, (u,v) \in E$ iff $(v,u) \in E$.

\noindent
\textbf{Sampling.} 
A subset $S \subseteq V$ of the population will be selected uniformly at random from $V$ and asked to report information about their in-neighbors in $G$. In the special case of \emph{full sampling,} we have that $S = V$. Each sampled vertex $v$ reports:
\begin{itemize}
    \item the number of in-neighbors $v$ has in $G$, denoted by $\reach{v}$, 
    \item the number of in-neighbors that belong to the hidden population $H$, denoted by $\cases{v}$.
\end{itemize} 
To preserve its privacy, a sampled vertex $v$ \emph{does not count itself in $\reach{v}$ and $\cases{v}$}. 
\begin{defn}
We define the following random variables. We assume the in-neighbors of $v$ numbered from 1 to $\reach{v}$ for easy reference.
\begin{itemize}
    \item
    For each vertex $v \in S$,  if the $j$-th in-neighbor of $v$, for $j \in \{1,2,\ldots,\reach{v}\}$, belongs to $H$, then the indicator random variable $X_{v_j} = 1$; otherwise, $X_{v_j} = 0$. Observe that $\cases{v}=\sum_{j=1}^{\reach{v}} X_{v_j}$.
    
    \item
    For each vertex $v \in S$, we define $Y_v \triangleq \cases{v} / \reach{v}$. 
\end{itemize}
\end{defn}
\noindent
\textbf{Problem definition.} An instance $I=(G,H)$ is a pair of a network and a hidden sub-population. The \emph{\prevalence} of instance $I$ is defined as $\ratio{(I)} \triangleq |H|/|V|$. $I[S]$ represents the set of pairs $\{ (C_v,R_v) :\ v\in S \}$ of the instance I obtained from the sample $S$.
\begin{defn}[Problem]
    Given an \emph{unknown} instance $I=(G,H)$ with $n=|V|$ and $h = |H|$, and a \emph{sampled} vertex set $S$, we are interested in finding an accurate 
    estimate $\estimate{}{(I[S])}$ of the \prevalence $\ratio{(I)}=h/n$.
\end{defn}

\noindent
\textbf{NSUM Estimators.}
The estimators studied in this work for estimating the \prevalence $\ratio{(I)}$ of instance $I=(G,H)$ are:
\begin{itemize}
    \item The \textit{Mean of Ratios} (\MoR) estimate is the average of the individual ratios of the sampled vertices, defined as 
    $$
    \estimate{\MoR}{(I[S])} \triangleq \frac{1}{|S|} \sum_{v \in S} \frac{\cases{v}}{\reach{v}} =  \frac{1}{|S|} \sum_{v \in S} Y_v.
    $$
    \item The \textit{Ratio of Sums} (\RoS) estimate is the ratio of the sum of the number of hidden neighbors over the sum of in-degrees of the sampled vertices, defined as 
    $$
    \estimate{\RoS}{(I[S])} \triangleq  \frac{\sum_{v \in S} \cases{v}}{\sum_{v \in S} \reach{v}} = \frac{\sum_{v \in S} \sum_{j=1}^{\reach{v}} X_{v_j}}{\sum_{v \in S} \reach{v}}.$$
\end{itemize}

\noindent
\textbf{Errors.} 
We define two different types of errors: upper error and lower error. Let $\method$ be any estimation method and $I=(G,H)$ an instance. Let $\estimate{\method}{(I[S])}$ be the estimate of $\ratio{(I)}$ obtained by $\method$.  We define  $\error{\method}{(I,S)}{+} \triangleq \max \paren{1,\frac{\estimate{\method}{(I[S])}}{\ratio{(I)}}}$ as the \textit{upper error} and $\error{\method}{(I,S)}{-} \triangleq \max \paren{ 1,\frac{\ratio{(I)}}{\estimate{\method}{(I[S])}}}$ as the \textit{lower error.} 
Note that $\error{\method}{(I,S)}{+}$ quantifies the factor by which $\estimate{\method}{(I[S])}$ is higher than $\ratio{(I)}$, and $\error{\method}{(I,S)}{-}$ quantifies the factor by which $\estimate{\method}{(I[S])}$ is lower than $\ratio{(I)}$.
Finally, we define the \textit{error of method $\method$ applied to instance $I$ and sample $S$} as 
\begin{eqnarray*}
\error{\method}{(I,S)}{} &\triangleq& \max \paren{\error{\method}{(I,S)}{+}, \error{\method}{(I,S)}{-}} \\ 
&=& \max \paren{\frac{\ratio{(I)}}{\estimate{\method}{(I[S])}},\frac{\estimate{\method}{(I[S])}}{\ratio{(I)}}}.
\end{eqnarray*}

We will denote $\error{\method}{(I)}{+}$, $\error{\method}{(I)}{-}$, and $\error{\method}{(I)}{}$ to the errors $\error{\method}{(I,V)}{+}$, $\error{\method}{(I,V)}{-}$, and $\error{\method}{(I,V)}{}$ corresponding to the full sampling case.

We say that $f(n,h)$ is an upper bound of the error of $\method$, i.e., $\error{\method}{}{} \leq f(n,h)$, if for all $I=(G=(V,E),H)$ such that $|V| = n$ and $|H| = h$, it holds that $\error{\method}{(I,S)}{} \leq f(n,h)$. Conversely, we say that the error of $\method$ is lower bounded by $f(n,h)$, i.e., $\error{\method}{}{} \geq f(n,h)$, if there exists $I=(G,H)$ and sample S such that $\error{\method}{(I,S)}{} \geq f(n,h)$. 

\noindent
\textbf{Lower Bound for Adversarial Instances.}
Let us consider any NSUM estimation method $\method$. We allow $\method$ to use all pairs $(\reach{v},\cases{v})$, for all $v\in V$, in any manner the method deems fit. We show that there are instances for which the error is $\error{\method}{}{} =\Omega(\sqrt{n})$. 

\begin{restatable}{theorem}{lowerany} \label{t-lower-any}
Any deterministic estimation method $\method$ that only uses the set $\{(\reach{v},\cases{v}), \forall v \in V\}$ to estimate the \prevalence $\ratio{}$ has error $\error{\method}{}{} \geq \sqrt{(n-1)/2}$.
\begin{proof}
We construct a bidirectional network $G=(V,E)$ and pair it with two different hidden population sets $H_1 \subset V$ and $H_2 \subset V$ resulting in instances $I_1=(G,H_1)$ and $I_2=(G,H_2)$. These two instances have two different prevalences $\ratio{(I_1)} = |H_1|/n$ and $\ratio{(I_2)} = |H_2|/n$. Our construction ensures that $\{(\reach{v},\cases{v}), \forall v \in V\}$ is the same for both choices of $H_1$ and $H_2$. Hence, the \textit{deterministic} method $\method$ should give the same estimate for both $I_1$ and $I_2$, i.e., $\estimate{\method}{(I_1[V])}=\estimate{\method}{(I_2[V])}$. We show that for all possible deterministic methods $\method$, this value is either far from $\ratio{(I_1)}$ or far from $\ratio{(I_2)}$.

The network $G$ is the union of a set of $k$ nodes (set $V_c$) forming a clique, a set of $k$ additional nodes (set $V_a$), each connected to a different clique node (and only to this node), and one single node $s$ connected to all the clique nodes $V_c$; see Fig. \ref{fig:thm1}. Observe that $n=2k+1$, and hence $k=\Theta(n)$. We consider the instances $I_1=(G,\{s\})$ and $I_2=(G,V_a)$. In $I_1$, the set of hidden nodes contains only node $s$, and hence the prevalence is $\ratio{(I_1)}=1/n=1/(2k+1)$. In $I_2$, all the $k$ nodes in $V_a$ belong to the hidden population, and hence the prevalence is $\ratio{(I_2)}=k/(2k+1)$.

\begin{figure} [ht] 
    \begin{center}
  \includegraphics[scale = 0.42]{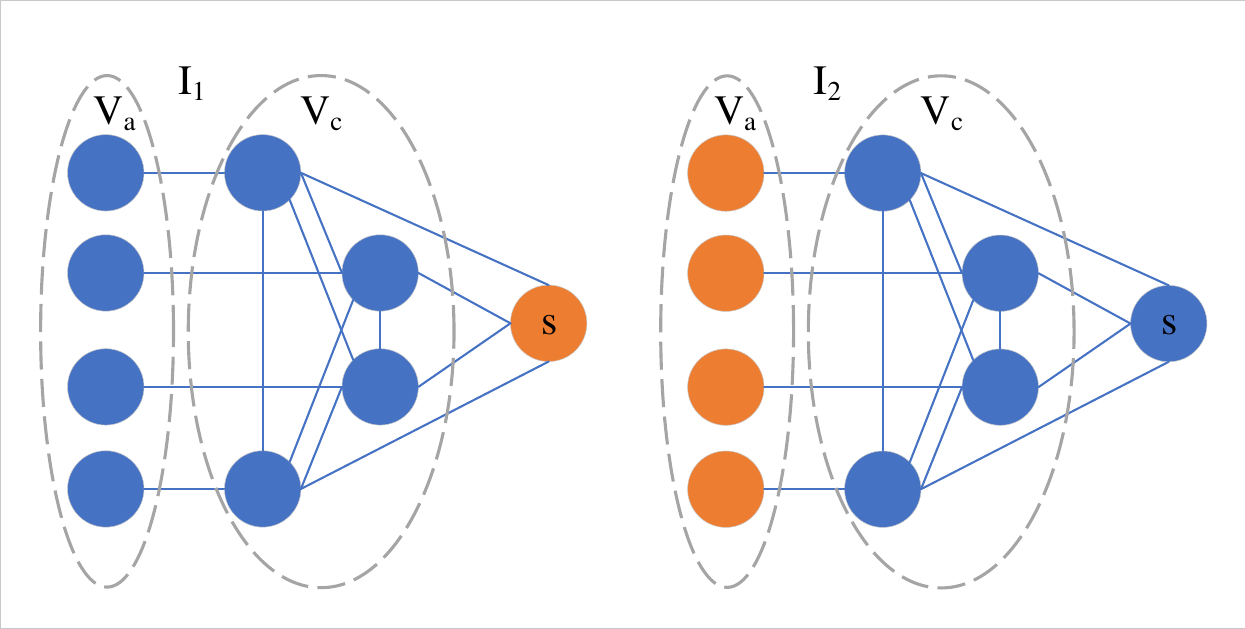}
    \end{center} \caption{Example instances $I_1$ and $I_2$ with $k=4$ and $n = 9$. All the edges are bidirectional, and the hidden nodes are in orange.} \label{fig:thm1}
    \Description{}
\end{figure}

However, observe that the pairs $(\reach{v},\cases{v})$ reported by all nodes are exactly the same in both $I_1$ and $I_2$: the graph $G$ is the same, $s$ and the nodes in $V_a$ have no hidden neighbor (and hence return $\cases{v}=0$), and in both cases, the nodes in $V_c$ have exactly one hidden neighbor (and hence return $\cases{v}=1$). Hence, it holds that $\estimate{\method}{(I_1)}=\estimate{\method}{(I_2)}$. Let us denote this estimate as $\estimate{\method}{}$ and its error as $\error{\method}{}{}$.

To bound the error, let us first consider $\estimate{\method}{} \notin [\ratio{(I_1)}),\ratio{(I_2)}]$. There are two cases:
if $\estimate{\method}{} < \ratio{(I_1)}$ then $\error{\method}{(I_2)}{} = \ratio{(I_2)}/\estimate{\method}{} > \ratio{(I_2)}/\ratio{(I_1)}=k$; if $\estimate{\method}{} > \ratio{(I_2)}$ then $\error{\method}{(I_1)}{} = \estimate{\method}{}/\ratio{(I_1)} > \ratio{(I_2)}/\ratio{(I_1)}=k$.
On the other hand, when $\estimate{\method}{} \in [\ratio{(I_1)}),\ratio{(I_2)}]$,
$$\error{\method}{(I_1)}{+} = \max\left(1,\frac{\estimate{\method}{}}{\ratio{(I_1)}}\right)=(2k+1)\ \estimate{\method}{},$$ 
and 
$$
\error{\method}{(I_2)}{-} = \max\left(1,\frac{\ratio{(I_2)}}{\estimate{\method}{}}\right) =\frac{k}{(2k+1)\estimate{\method}{}}.
$$
The error incurred by method $\method$ is lower bounded by the maximum of $\error{\method}{(I_1)}{+}$ and $\error{\method}{(I_2)}{-}$ for both $I_1$ and $I_2$. This maximum value is minimized when $\estimate{\method}{}=\sqrt{k}/(2k+1)$, which implies that $\error{\method}{}{} \geq \max(\error{\method}{(I_1)}{+}, \error{\method}{(I_2)}{-}) \geq \sqrt{k}$. From our construction $k = (n-1)/2$, and the error is at least $\error{\method}{}{} \geq \sqrt{k}=\sqrt{(n-1)/2}$.
\end{proof}
\end{restatable}

The lower bound in Thm.~\ref{t-lower-any} is surprising because it is valid for the full sampling scenario where $S = V$. This raises the question regarding the goodness of the estimates provided by $\MoR$ and $\RoS$, which only use $\{(\reach{v},\cases{v}), \forall v \in S\}$. We address this question and show that the error can be bounded for these methods for random networks, with a small number of samples. We include two popular classes of random networks, Scale-Free \cite{clauset2009power} and Erdős–Rényi networks \cite{erdosrenyi}.

\noindent
\textbf{Random networks.} 
In the sequel, we study instances $I=(G,H)$ in which $G$ is a random network over a fixed population $V$ and a given hidden sub-population $H \subseteq V$. To be precise, the in-neighbors of each vertex $v \in V$ of the network $G$ are selected independently as follows:
\begin{itemize}
    \item Vertex $v$ selects a degree $\reach{v}$ from a given degree distribution $P_\text{deg}$.
    \item Given $\reach{v}$, vertex $v$ selects $\reach{v}$ vertices uniformly at random from $V \setminus \{v\}$ to be its set of in-neighbors. 
\end{itemize}
This construction is general in that we build different graph families by choosing different degree distributions. As special cases of these general random networks, we study two popular classes of random networks, Scale-Free \cite{clauset2009power} and Erdős–Rényi networks \cite{erdosrenyi}.
A directed Scale-Free network is a random graph in which the fraction of nodes with in-degree $k$ is proportional to $k^{-\gamma}$, where $\gamma$ is a small constant (typically, $\gamma \in (2,3)$).
A directed random graph generated with the Erdős–Rényi model $G(n, p)$ has $n$ nodes, and for any two nodes $u,v \in V$, the edge $(u,v)$ is included in the edge set $E$ with probability $p$, independently of any other edge. 

Note that in random networks, the random variables $Y_v$, for all $v \in S$, are dependent. Likewise, $X_{j_v}$, for all $v \in S$ and $j \in \{1,2,\ldots,\reach{v}\}$, are also dependent. Thus, the direct application of concentration bounds to obtain upper bounds for the errors is not possible. We address this challenge by using conditional independence in the random networks and showing a negative correlation for the above variables. 




\section{Upper Bounds on the Error of \MoR and \RoS in Random Networks}\label{sec:mainresults}

For ease in exposition and conciseness in the expressions, in this section, we prove results for random networks in which the support of the degree distribution $P_\text{deg}$ is positive integers, i.e., the degree of any node is non-zero. We state it in the following assumption:
\begin{assumption}\label{non-zero-degree-assumption}
    $R_v \geq 1$ for all $v \in V$.
\end{assumption}
For random networks without the above assumption, but with bounded probability of zero degree, one can modify the error bounds derived in this section to accommodate those scenarios. The omitted proofs can be found in the Appendix.


\subsection{Upper Bound on the Error of the \MoR Estimator}
Recall that $\estimate{\MoR}{(I[S])} = \frac{1}{|S|} \sum_{v \in S} \frac{\cases{v}}{\reach{v}} = \frac{1}{|S|} \sum_{v \in S} Y_v$. We first show that, for random networks, $\estimate{\MoR}{(I[S])}$ is an \textit{unbiased estimate} of $\ratio{(I)}$, i.e., $\E[\estimate{\MoR}{(I[S])}] = \E[Y_v]=\ratio{(I)}$.
Since a vertex $v$ cannot select itself as an in-neighbor (i.e., there is no self-reporting), the distribution of the random variables $X_{vj}$ depends on whether $v$ is in $H$. Using combinatorial arguments, we show that $\E[X_{v_j} \mid v \in H] = \frac{h-1}{n-1}$, and $\E[X_{v_j} \mid v \notin H] = \frac{h}{n-1}$, for all $j \in \{1,2,\ldots R_v\}$.
%
%
%
From this result, we obtain that the expectation of $X_{vj}$ and $Y_{v}$ are equal to $\ratio{(I)}$.



\begin{restatable}{lemma}{lemexpXYv} \label{lem:exp-Y-v}
For all $v \in S$ and $j \in \{1,2,\ldots,R_v\}$, $\E[X_{v_j}]=\E[Y_{v}] = \frac{h}{n} = \ratio{(I)}$.
\end{restatable}


We note that the variables $Y_v$ are not independent if we do not know the membership of $v$ in $H$. Consider, for example, the case where $|H|=1$ and two vertices $v_1, v_2 \in S$. Then,
$$
\Pr[Y_{v_2}>0 \mid Y_{v_1} >0] < \Pr[Y_{v_2}>0],
$$
because if $Y_{v_1} >0$, it is sure that $v_1 \notin H$. This increases the probability of $v_2 \in H$, reducing the probability that $Y_{v_2}>0$.
%
However, we show that variables $Y_v$, for all $v \in S$, are negatively correlated.
%
%
Using this and the fact that $\E[Y_v]=\ratio{(I)}$, we apply a concentration bound for negatively correlated random variables to obtain an upper bound for the error stated in the following theorem.
%
Let us define the following function
$$
F(x,y) = \paren{\frac{e^{x-1}}{x^{x}}}^{y} + \paren{\frac{e^{\frac{1}{x}-1}}{x^{-1/x}}}^{y}.
$$
\begin{restatable}{theorem}{thmMoRboundmain} \label{thm:MoR-bound-main}
Consider an instance $I=(G,H)$ where $G$ is a random network. If $S \subseteq V$ is sampled uniformly at random, with $m=|S|$, then the \MoR estimator has the following error bound for any $\beta=1+\epsilon > 1$.
\begin{align*}
 &\Pr\brac{\error{\MoR}{(I,S)}{} > \beta } \nonumber 
  \leq 
 F(\beta,m \ratio{(I)}).
\end{align*}
\end{restatable}

\subsection{Upper Bound on the Error with the \RoS Estimator}
For the \RoS estimator, we define the random variables $\Reach{S}=\sum_{v \in S} \reach{v}$ and $\Cases{S}=\sum_{v \in S} \cases{v}=\sum_{v \in S} \sum_{j=1}^{\reach{v}} X_{vj}$. By definition, the estimator is 
$\estimate{\RoS}{(I[S])}=\Cases{S} / \Reach{S}$.
It was already shown in \Cref{lem:exp-Y-v} that the expectation of $X_{vj}$ is equal to $\ratio{(I)}$. We can also show that the variables $X_{vj}$ are negatively correlated.
%
%
%
Then, we can use a concentration bound to prove the following theorem.
%
\begin{restatable}{theorem}{thmRoSboundmain} \label{thm:RoS-bound-main}
Consider an instance $I=(G,H)$ where $G$ is a random network. If $S \subseteq V$ is sampled uniformly at random, with $m=|S|$, then for $\Reach{S}=\sum_{v \in S} \reach{v}$ and any $\beta=1+\epsilon > 1$, the \RoS estimator has the following error bound.
$$
\Pr\brac{\error{\RoS}{(I,S)}{} > \beta } 
  \leq \\ \sum_{\Reach{}} 
   F(\beta,\Reach{} \ratio{(I)})
  \Pr\Brack{\Reach{S}=\Reach{}}.
$$
\end{restatable}


\begin{corollary}
\label{Cor:RoS-bound-main}
Consider an instance $I=(G,H)$ where $G$ is a random network $G$. If $S \subseteq V$ is sampled uniformly at random, with $m=|S|$, then the \RoS estimator has the following error bound, for any $\beta=1+\epsilon > 1$.
\begin{align*}
 &\Pr\brac{\error{\RoS}{(I,S)}{} > \beta } \leq 
 F(\beta, m\ratio{(I)}).
\end{align*}
\end{corollary}
This bound matches the bound for \MoR in Thm.~\ref{thm:MoR-bound-main}. Since, in general, we expect $\Reach{S} > |S|$, the bound for \RoS in Thm.~\ref{thm:RoS-bound-main} is usually better than the bound in Thm.~\ref{thm:MoR-bound-main} for \MoR.

\subsection{Sample Size}
From Thm.~\ref{thm:MoR-bound-main} and Cor.~\ref{Cor:RoS-bound-main}, we can show that it is possible to achieve an error of at most $\beta=1+\epsilon>1$ with high probability using a logarithmic sample set.

\begin{restatable}{theorem}{thmsamplesize} \label{thm:sample-size}
Given an instance $I=(G,H)$ where $G$ is a random network. If $G$ is sampled uniformly at random, 
with 
$$m=|S|\geq \frac{\ln{2} + \alpha \ln{n}}{\ratio{(I)} (1 -\frac{1}{\beta} (\ln{\beta} + 1))},$$ 
then the error of \MoR and \RoS estimators is no larger than $\beta=1+\epsilon > 1$ with
probability at least $1-1/n^\alpha$, for any constant $\alpha>0$.
\end{restatable}

\begin{proof}
We show that $\left(\frac{e^{\beta-1}}{\beta^{\beta}}\right)^{m\ratio{(I)}} \leq 1/(2n^\alpha)$ and $\left(\frac{e^{\frac{1}{\beta}-1}}{\beta^{-1/\beta}}\right)^{m\ratio{(I)}} \leq 1/(2n^\alpha)$, which implies the claim from \Cref{thm:MoR-bound-main} and \Cref{Cor:RoS-bound-main}.
Hence, $m$ must simultaneously satisfy
\begin{align*}
\left(\frac{e^{\beta-1}}{\beta^{\beta}}\right)^{m\ratio{(I)}} &\leq 1/(2n^\alpha) \\
m\ratio{(I)} \ln{\left(\frac{e^{\beta-1}}{\beta^{\beta}}\right)} &\leq - (\ln{2} + \alpha \ln{n})\\
m \geq \frac{\ln{2} + \alpha \ln{n}}{\ratio{(I)} \ln{\left(\frac{\beta^{\beta}}{e^{\beta-1}}\right)}} &= \frac{\ln{2} + \alpha \ln{n}}{\ratio{(I)} (\beta \ln{\beta} - \beta +1)},
\end{align*}
and 
\begin{align*}
\left(\frac{e^{\frac{1}{\beta}-1}}{\beta^{-1/\beta}}\right)^{m\ratio{(I)}} &\leq 1/(2n^\alpha) \\
m\ratio{(I)} \ln{\left(\frac{e^{\frac{1}{\beta}-1}}{\beta^{-1/\beta}}\right)} &\leq - (\ln{2} + \alpha \ln{n})\\
m \geq \frac{\ln{2} + \alpha \ln{n}}{\ratio{(I)} \ln{\left(\frac{\beta^{-1/\beta}}{e^{\frac{1}{\beta}-1} }\right)}} &= \frac{\ln{2} + \alpha \ln{n}}{\ratio{(I)} (-\frac{1}{\beta} \ln{\beta} - \frac{1}{\beta} +1)},
\end{align*}
It holds that $\beta \ln{\beta} - \beta +1 \geq -\frac{1}{\beta} \ln{\beta} - \frac{1}{\beta} +1 >0$ for $\beta>1$. Hence, the bound on $m$ satisfies that 
$\left(\frac{e^{\beta-1}}{\beta^{\beta}}\right)^{m\ratio{(I)}} \leq 1/(2n^\alpha)$ and $\left(\frac{e^{\frac{1}{\beta}-1}}{\beta^{-1/\beta}}\right)^{m\ratio{(I)}} \leq 1/(2n^\alpha)$, and hence $\Pr\left[\error{}{(I,S)}{} > \beta \right] \leq 1/n^\alpha$.
\end{proof}

The bound on the sample size provided by this theorem is $O(\log n)$ as long as $\alpha$, $\beta$, and $\ratio{(I)}$ are constant. Observe that $\alpha$ and $\beta$ can be freely chosen, but the \prevalence $\ratio{(I)}$ is not known. Instead, a lower bound on $\ratio{(I)}$ has to be assumed and used to obtain the size $m$ in practice.
\section{Upper Bounds for Erdős–Rényi and Scale-Free Networks}
\label{sec:random-graphs}


\subsection{Erdős–Rényi Network\label{erdos-renyi-sub-section}}

Consider an instance $I=(G,H)$ where the directed network $G = (V, E)$ is generated with the Erdős–Rényi model $G(n, p)$. In this network, a directed link $(u,v)$ between nodes $u,v \in V$ is created independently at random with probability $p$. This is equivalent to assigning to each vertex $v \in V$ an in-degree $\reach{v} \sim Binomial(n-1,p)$ and selecting the in-neighbors of $v$ uniformly at random from $V \setminus \{v\}$.
Unfortunately, Assumption \ref{non-zero-degree-assumption} does not hold for Erdős–Rényi networks since there may be vertices with in-degree zero. To prevent this, we assume that zero degree nodes are not allowed, and distribute the probability of that event among the rest of in-degrees proportionally.

\noindent
\textbf{The bound for the \MoR estimator.}
%
We can directly apply \Cref{thm:MoR-bound-main}.

\begin{restatable}{theorem}{thmMoRerdosrenyi} \label{thm:MoR-erdos-renyi}
Consider an instance $I=(G,H)$ where the directed random network $G = (V, E)$ is generated with the Erdős–Rényi model $G(n, p)$ without zero in-degree nodes. If $S \subseteq V$ is sampled uniformly at random, with $m=|S|$, then the \MoR estimator has the following error bound, for any $\beta=1+\epsilon > 1$.
$$
\Pr\brac{\error{\MoR}{(I,S)}{} > \beta }
  \leq 
  F(\beta, m\ratio{(I)})
$$
\end{restatable}

\noindent
\textbf{The bound for the \RoS estimator.} Let us now consider the \RoS estimator. Observe that by construction of the Erdős–Rényi model $G(n, p)$ without zero in-degree nodes the average in-degree is no less than $p(n-1)$.

\begin{restatable}{theorem}{thmRoSerdosrenyi} \label{thm:RoS-erdos-renyi}
Consider an instance $I=(G,H)$ where the directed random network $G = (V, E)$ is generated with the Erdős–Rényi model $G(n, p)$ without zero in-degree nodes, and a set $S \subseteq V$ is sampled uniformly at random, with $m=|S|$. 
Then, the \RoS estimator has the following error bound, for any $\delta \in (0,1)$ and any $\beta=1+\epsilon > 1$ and where $\mu = m p(n-1)$.
$$
\Pr\brac{\error{\RoS}{(I,S)}{} > \beta } 
  \leq \paren{\frac{e^{-\delta}}{(1-\delta)^{(1-\delta)}}}^{\mu} 
  + F(\beta, (1-\delta) \mu \ratio{(I)}).
$$
\end{restatable}

\subsection{Scale-Free Network\label{scale-free-sub-section}}

In a Scale-Free random network $G = (V, E)$, the probability that a node has in-degree $k \geq 1$, independent of other nodes, is given by a degree distribution $\pdeg{k} = \nu k^{-\gamma}$, where $\nu=1/\sum_{k=1}^{n-1} k^{-\gamma}$ is a constant that ensures that $\sum_{k=1}^{n-1} \pdeg{k} = 1$. Note that for this network, there are no nodes with in-degree zero, and hence the Assumption \ref{non-zero-degree-assumption} is true, and the results in Section \ref{sec:mainresults} can be directly applied.
We adopt the Scale-Free network model in which for any vertex $v$ with in-degree $k$, the $k$ in-neighbors of $v$ are picked uniformly at random from $|V| \setminus \{v\}$.
Let us consider an instance $I=(G,H)$ in which the random network $G$ is a scale-network as described above. 

\noindent
\textbf{The bounds for the \MoR estimator.} For the \MoR estimator, the error bound we obtain is the one provided by Thm.~\ref{thm:MoR-bound-main}.

\noindent
\textbf{The bound for the \RoS estimator.} 
For the \RoS estimator, the error bound provided by Cor.~\ref{Cor:RoS-bound-main} also applies. However, in this case, we obtain a better bound by using Thm.~\ref{thm:RoS-bound-main}. Recall that Scale-Free networks usually have $\gamma \in (2,3)$. The result below assumes $\gamma>2$ to avoid a negative value of $\mu$.

\begin{restatable}{theorem}{thmRoSbound} \label{thm:RoS-bound}
Consider an instance $I=(G,H)$ in which $G$ is a Scale-Free random network with $\gamma>2$. If $S \subseteq V$ is sampled uniformly at random, with $m=|S|$, 
then the \RoS estimator has the following error bound, for  any $\delta \in (0,1)$, and any $\beta=1+\epsilon > 1$ and where $\mu \approx m \frac{1-\gamma}{2-\gamma} \frac{1-(n-1)^{2-\gamma}}{1-(n-1)^{1-\gamma}}$.
$$
\Pr\brac{\error{\RoS}{(I,S)}{} > \beta }
  \leq \paren{\frac{e^{-\delta}}{(1-\delta)^{(1-\delta)}}}^{\mu} + 
  F(\beta, (1-\delta) \mu \ratio{(I)}).
$$
\end{restatable}
   
\definecolor{mycolor1}{rgb}{0.00000,0.44700,0.74100}
\definecolor{mycolor2}{rgb}{0.85000,0.32500,0.09800}
\definecolor{mycolor3}{rgb}{0.92900,0.69400,0.12500}
\definecolor{mycolor4}{rgb}{0.4666,0.67451,0.18824}
\definecolor{mycolor5}{rgb}{0.49412,0.18431,0.30588}
\definecolor{mycolor6}{rgb}{0.30196,0.69020,0.93333}

\begin{figure}[t]
\begin{center}
    \definecolor{mycolor1}{rgb}{0.00000,0.44700,0.74100}
\definecolor{mycolor2}{rgb}{0.85000,0.32500,0.09800}
\definecolor{mycolor3}{rgb}{0.92900,0.69400,0.12500}
\definecolor{mycolor4}{rgb}{0.4666,0.67451,0.18824}
\definecolor{mycolor5}{rgb}{0.49412,0.18431,0.30588}
\definecolor{mycolor6}{rgb}{0.30196,0.69020,0.93333}

\begin{tabular}{c c c}
\hspace{-17pt}
\begin{minipage}{0.335\linewidth}
    \resizebox{1.00\textwidth}{!}{
    \begin{tikzpicture}
    \begin{axis}[
        scale = 0.7,
        grid=major,
	xlabel = \Large Sample Size $|S|$,
        title = \LARGE (a) $\mathcal{E}_{\mathcal{M}}$,
        boxplot/draw direction=y,
        boxplot={
        draw position={1/4 + floor(\plotnumofactualtype/2) + 1/2*mod(\plotnumofactualtype,2)},
        box extend=0.25,
        },
        legend style={
          cells={anchor=west}, 
          legend image code/.code={ 
            \draw[#1] (0cm,-0.1cm) -- (0.3cm,-0.1cm); 
          }
        },
  	ymin = 1,
  	ymax = 1.056,
        xtick={0,1,2,...,50},
        x tick label as interval,
        xticklabels={%
      {$10^2$},%
      {$5~10^2$},%
      {$10^3$},%
      {$5~10^3$},%
      {$10^4$},%
      {$5~10^4$},%
        },
  	ticklabel style = {font=\large},
  	ytick={1.00, 1.05,1.1, 1.2},
   	]

   \draw[dashed, line width=2pt, black] 
      (axis cs:0,1.02) -- 
      (axis cs:50,1.02) 
      node[pos=1, anchor=west, font=\large] {\textcolor{red}{Reference}};

   \draw[dashed, line width=2pt, black] 
      (axis cs:0,1.01) -- 
      (axis cs:50,1.01) 
      node[pos=1, anchor=west, font=\large] {\textcolor{red}{Reference}};

    \node[anchor=south west, font=\large] at (axis cs:4,1.02) {\textbf{\LARGE $\epsilon = 2\%$}};

    \node[anchor=south west, font=\large] at (axis cs:4,1.01) {\textbf{\LARGE $\epsilon = 1\%$}};

\addplot[line width = 2.00pt, mycolor1]
table[row sep=\\,y index=0] {
data\\
1.000000\\
1.000000\\
1.000000\\
1.054576\\
1.295237\\
};
\addlegendentry{\Large $\mathcal{E}_{MoR}$}

\addplot[line width = 2.00pt, mycolor2]
table[row sep=\\,y index=0] {
data\\
1.000000\\
1.000000\\
1.000000\\
1.054054\\
1.318681\\
};
\addlegendentry{\Large $\mathcal{E}_{RoS}$}

\addplot[line width = 2.00pt, mycolor1]
table[row sep=\\,y index=0] {
data\\
1.000000\\
1.000000\\
1.000000\\
1.024171\\
1.141020\\
};

\addplot[line width = 2.00pt, mycolor2]
table[row sep=\\,y index=0] {
data\\
1.000000\\
1.000000\\
1.000000\\
1.023848\\
1.140841\\
};
\addplot[line width = 2.00pt, mycolor1]
table[row sep=\\,y index=0] {
data\\
1.000000\\
1.000000\\
1.000142\\
1.017287\\
1.097095\\
};

\addplot[line width = 2.00pt, mycolor2]
table[row sep=\\,y index=0] {
data\\
1.000000\\
1.000000\\
1.000134\\
1.016986\\
1.091725\\
};

\addplot[line width = 2.00pt, mycolor1]
table[row sep=\\,y index=0] {
data\\
1.000000\\
1.000000\\
1.000000\\
1.007714\\
1.043535\\
};

\addplot[line width = 2.00pt, mycolor2]
table[row sep=\\,y index=0] {
data\\
1.000000\\
1.000000\\
1.000000\\
1.007558\\
1.041978\\
};
\addplot[line width = 2.00pt, mycolor1]
table[row sep=\\,y index=0] {
data\\
1.000000\\
1.000000\\
1.000040\\
1.005434\\
1.034558\\
};

\addplot[line width = 2.00pt, mycolor2]
table[row sep=\\,y index=0] {
data\\
1.000000\\
1.000000\\
1.000070\\
1.005386\\
1.031932\\
};

\addplot[line width = 2.00pt, mycolor1]
table[row sep=\\,y index=0] {
data\\
1.000000\\
1.000000\\
1.000000\\
1.002407\\
1.017379\\
};

\addplot[line width = 2.00pt, mycolor2]
table[row sep=\\,y index=0] {
data\\
1.000000\\
1.000000\\
1.000000\\
1.002381\\
1.017314\\
};

  	\end{axis}
  \end{tikzpicture}
 }
\end{minipage}

&
\hspace{-12pt}
\begin{minipage}{0.335\linewidth}
    \resizebox{1.00\textwidth}{!}{
    \begin{tikzpicture}
  	\begin{axis}[
    scale = 0.7,
  	legend pos = north east,
   	grid = major,
	xlabel = \LARGE Sample Size $|S|$,
    title = \LARGE (b) $\mathcal{E}_{MoR}$ (\textrm{$\rho=0.05$}),
  	ymin = 1,
  	ymax = 1.056,
  	xmin = 100,
 	xmax = 500000,
    xmode = log,
  	ticklabel style = {font=\Large},
  	ytick={1.00, 1.05,1.1, 1.2},
   	]
  	\addplot[solid, smooth, line width = 3.00pt, mycolor3] table[x ={sizes}, y = eMoR]{sub_parts/FigData/mean_eMoR_deg_15_rho_5.dat};
    \addlegendentry{\Large $p(n-1)=15$}
  	\addplot[solid, smooth, line width = 3.00pt, mycolor2] table[x ={sizes}, y = eMoR]{sub_parts/FigData/mean_eMoR_deg_30_rho_5.dat};
    \addlegendentry{\Large $p(n-1) = 30$}
  	\addplot[solid, smooth, line width = 3.00pt, mycolor1] table[x ={sizes}, y = eMoR]{sub_parts/FigData/mean_eMoR_deg_50_rho_5.dat};
    \addlegendentry{\Large $p(n-1)=50$}
    
   \node[] at (axis cs: 100000,1.11) {\textbf{$\pmb{\epsilon = 10\%}$}};
   \node[] at (axis cs: 100000,1.023) {\textbf{\LARGE $\epsilon = 2\%$}};
   \node[] at (axis cs: 100000,1.013) {\textbf{\LARGE $\epsilon = 1\%$}};
   \addplot[dashed, line width = 2.00pt] (100,1.02) -- (1000000,1.02);
   \addplot[dashed, line width = 2.00pt] (100,1.01) -- (1000000,1.01);
   \addplot[dashed, line width = 1.00pt] (100,1.10) -- (1000000,1.10);
  	\end{axis}
  \end{tikzpicture}
 }
\end{minipage}
&
\hspace{-12pt}
\begin{minipage}{0.335\linewidth}
    \resizebox{1.00\textwidth}{!}{
    \begin{tikzpicture}
  	\begin{axis}[
     scale = 0.70,
  	legend pos = north east,
   	grid = major,
	xlabel = \LARGE Sample Size $|S|$,
    title = \LARGE (c) $\mathcal{E}_{RoS}$ (\textrm{$\rho=0.05$}),
  	ymin = 1,
  	ymax = 1.056,
  	xmin = 100,
 	xmax = 500000,
    xmode = log,
  	ticklabel style = {font=\Large},
  	ytick={1.00, 1.05, 1.1, 1.2},
   	]
  	\addplot[solid, smooth, line width = 3.00pt, mycolor3] table[x ={sizes}, y = eRoS]{sub_parts/FigData/mean_eMoR_deg_15_rho_5.dat};
    \addlegendentry{\Large $p(n-1)=15$}
  	\addplot[solid, smooth, line width = 3.00pt, mycolor2] table[x ={sizes}, y = eRoS]{sub_parts/FigData/mean_eMoR_deg_30_rho_5.dat};
    \addlegendentry{\Large $p(n-1) = 30$}
  	\addplot[solid, smooth, line width = 3.00pt, mycolor1] table[x ={sizes}, y = eRoS]{sub_parts/FigData/mean_eMoR_deg_50_rho_5.dat};
    \addlegendentry{\Large $p(n-1)=50$}
   \node[] at (axis cs: 100000,1.11) {\textbf{$\pmb{\epsilon = 10\%}$}};
   \node[] at (axis cs: 100000,1.023) {\textbf{\LARGE $\epsilon = 2\%$}};
   \node[] at (axis cs: 100000,1.013) {\textbf{\LARGE $\epsilon = 1\%$}};
   \addplot[dashed, line width = 2.00pt] (100,1.02) -- (1000000,1.02);
   \addplot[dashed, line width = 2.00pt] (100,1.01) -- (1000000,1.01);
   \addplot[dashed, line width = 1.00pt] (100,1.10) -- (1000000,1.10);
  	\end{axis}
  \end{tikzpicture}
 }
\end{minipage}
\end{tabular}
\end{center}
    \caption{Erdős–Rényi ($n = 10^{6}$). (a) Boxplots $\error{\method}{}{}$ versus $|S|$ for MoR and RoS for $\ratio{}=0.05$ and $p(n-1)=30$. (b) Average $\error{\MoR}{}{}$ versus $|S|$ for $\ratio{}=0.05$ and different $p(n-1)$. (c) Average $\error{\RoS}{}{}$ versus $|S|$ for $\ratio{}=0.05$ and different $p(n-1)$.}
    \label{fig_ER_box}
    \Description{}
\end{figure}


\begin{figure*}[t]
\begin{center}
    \input{sub_parts/FigTex/ErdosRenyi_bounds}
\end{center}
    \caption{Erdős–Rényi ($n = 10^{6}$). Simulation curves and bounds for (\ref{line_ER_d15}) $p(n-1)=15$, (\ref{line_ER_d30}) $p(n-1)=30$, and (\ref{line_ER_d50}) $p(n-1)=50$. (a) $\Pr[\error{\MoR}{}{} > 1 + \epsilon]$ versus $|S|$ for $\epsilon = 0.05$ and $\ratio{}$. (b) Sample size versus $\epsilon$ for $\ratio{} = 0.05$. (c) Sample size versus $\ratio{}$ for $\epsilon = 0.05$. (d) $\Pr[\error{\RoS}{}{} > 1 + \epsilon]$ versus $|S|$ for $\epsilon = 0.05$ and $\ratio{}=0.05$. (e) Sample size versus $\epsilon$ for $\ratio{} = 0.05$. (f) Sample size versus $\ratio{}$ for $\epsilon = 0.05$.}
    \label{fig_ER_bounds}
    \Description{}
\end{figure*}


\section{Simulations}

We evaluate the behavior of \MoR and \RoS estimators under different conditions using two network models: Erdős–Rényi and Scale-Free. For a given $\ratio{}$, we have generated 100 instances $I=(G,H)$ with that \prevalence for each network model. Then, we obtained 200 sample sets with a particular size $|S|$ from each instance. For each graph instance and sample set, we estimated \prevalence with $\estimate{\MoR}{(I)}$ and $\estimate{RoS}{(I)}$, and computed the respective errors $\error{\MoR}{(I)}{}$ and $\error{\RoS}{(I)}{}$. Simulation codes were implemented in MATLAB and executed on a Dell Inspiron 14 7000 laptop with Intel Core i7 CPU 2.80 GHz, RAM 16 GB, and Ubuntu 22.04 OS.

\subsection{Erdős–Rényi} 

We assess the errors $\error{\method}{}{}$ using the Erdős–Rényi network model. Each generated graph has a size of $n=10^6$ and a probability parameter $p$, where the mean in-degree is given by $p(n-1)$. Specifically, we evaluate the performance of the estimators for different mean in-degrees.

\noindent
\textbf{Error versus sample size.} We analyze the behavior of $\error{\method}{}{}$ in relation to the sample size $|S|$.
Fig.~\ref{fig_ER_box}(a) displays the boxplots of $\error{\MoR}{}{}$ and $\error{\RoS}{}{}$ as function of the sample size $|S|$ for a mean in-degree of $p(n-1) = 30$ and $\ratio{} = 0.05~(5\%)$. These boxplots indicate that $\error{\MoR}{}{}$ and $\error{\RoS}{}{}$ decrease as $|S|$ increases. For $|S| > 500$, most estimated errors are below $1.05~(\epsilon=5\%)$. 
Fig.~\ref{fig_ER_box}(b) and \ref{fig_ER_box}(c) show the average $\error{\MoR}{}{}$ and $\error{\RoS}{}{}$, respectively, versus $|S|$ for different mean in-degrees $p(n-1) = 15$, $30$, and $50$, with $\ratio{} = 0.05$. These figures include the threshold levels for $\epsilon=1\%$ and $2\%$ for easy reference. Average errors decrease exponentially with increasing sample size. In addition, higher mean in-degrees lead to lower errors, showing that the estimation accuracy improves with the increasing number of neighbors. Notice that the average $\error{\MoR}{}{}$ and $\error{\RoS}{}{}$ perform similarly for the Erdős–Rényi graph. In summary, the findings indicate that larger sample size, higher hidden population \prevalence, and higher mean in-degree improve accuracy in estimating hidden population \prevalence in an Erdős–Rényi network. 

\noindent
\textbf{Analysis of $\error{\MoR}{}{}$.} Then, we assess the behavior of $\error{\MoR}{}{}$ using the Erdős–Rényi model. Fig.~\ref{fig_ER_bounds}(a) shows the probability $\Pr[\error{\MoR}{}{} > 1 + \epsilon]$ versus $|S|$ for $\epsilon =0.05$, $\ratio{} = 0.05$, and $p(n-1) = 15$, $30$, and $50$. This figure compares simulation results with the analytical bounds from Thm.~\ref{thm:sample-size} and \ref{thm:MoR-erdos-renyi}. The bound derived in Thm.~\ref{thm:sample-size} is obtained using $\alpha = 1/2$. Fig.~\ref{fig_ER_bounds}(b) illustrates the sample size required to achieve $\Pr[\error{\MoR}{}{} > 1 + \epsilon] = 0.05$ versus $\epsilon$ for $\ratio{}=0.05$ and different mean in-degrees. Fig.~\ref{fig_ER_bounds}(c) displays the sample size required to achieve $\Pr[\error{\MoR}{}{} > 1 + \epsilon] = 0.05$ versus $\ratio{}$ for $\epsilon=0.05$ and various mean in-degrees. The required sample size decreases as \ratio{} increases, indicating that higher \ratio{} leads to more accurate estimates. 
Figs~\ref{fig_ER_bounds}(a)-(c) indicate that estimator performance improves with increasing mean in-degree. A higher mean in-degree result in better coverage of the target population; therefore, more information is available to estimate \ratio{}. Finally, the bounds from theorems \ref{thm:sample-size} and \ref{thm:MoR-erdos-renyi} are conservative concerning the simulation results.

\noindent
\textbf{Analysis of $\error{\RoS}{}{}$.} We examine the behavior of $\error{\RoS}{}{}$ using the Erdős–Rényi model. To this end, we consider the analytical bound in Thm.~\ref{thm:RoS-erdos-renyi}. Fig.~\ref{fig_ER_bounds}(d)-(f) illustrate the behavior of $\error{\RoS}{}{}$ in various scenarios. The bound derived in Thm.~\ref{thm:RoS-erdos-renyi} depends on the mean in-degree, and the resulting bounds are closer to the simulation results than those obtained by $\error{\MoR}{}{}$.

\begin{figure}[t]
\begin{center}
    \definecolor{mycolor1}{rgb}{0.00000,0.44700,0.74100}
\definecolor{mycolor2}{rgb}{0.85000,0.32500,0.09800}
\definecolor{mycolor3}{rgb}{0.92900,0.69400,0.12500}
\definecolor{mycolor4}{rgb}{0.4666,0.67451,0.18824}
\definecolor{mycolor5}{rgb}{0.49412,0.18431,0.30588}
\definecolor{mycolor6}{rgb}{0.30196,0.69020,0.93333}

\begin{tabular}{c c c}
\hspace{-15pt}
\begin{minipage}{0.335\linewidth}
    \resizebox{1.00\textwidth}{!}{
    \begin{tikzpicture}
    \begin{axis}[
        scale = 0.7,
        grid=major,
	xlabel = \Large Sample Size $|S|$,
        title = \LARGE (a) $\mathcal{E}_{\mathcal{M}}$,
        boxplot/draw direction=y,
        boxplot={
        draw position={1/4 + floor(\plotnumofactualtype/2) + 1/2*mod(\plotnumofactualtype,2)},
        box extend=0.25,
        },
        legend style={
          cells={anchor=west}, 
          legend image code/.code={ 
            \draw[#1] (0cm,-0.1cm) -- (0.3cm,-0.1cm); 
          }
        },
  	ymin = 1,
  	ymax = 1.12,
        xtick={0,1,2,...,50},
        x tick label as interval,
        xticklabels={%
      {$10^2$},%
      {$5~10^2$},%
      {$10^3$},%
      {$5~10^3$},%
      {$10^4$},%
      {$5~10^4$},%
        },
  	ticklabel style = {font=\large},
  	ytick={1.00, 1.05,1.1, 1.2},
   	]

   \draw[dashed, line width=2pt, black] 
      (axis cs:0,1.02) -- 
      (axis cs:50,1.02) 
      node[pos=1, anchor=west, font=\large] {\textcolor{red}{Reference}};

   \draw[dashed, line width=2pt, black] 
      (axis cs:0,1.05) -- 
      (axis cs:50,1.05) 
      node[pos=1, anchor=west, font=\large] {\textcolor{red}{Reference}};

    \node[anchor=south west, font=\large] at (axis cs:4,1.02) {\textbf{\LARGE $\epsilon = 2\%$}};

    \node[anchor=south west, font=\large] at (axis cs:4,1.05) {\textbf{\LARGE $\epsilon = 5\%$}};

\addplot[line width = 2.00pt, mycolor4]
table[row sep=\\,y index=0] {
data\\
1.000000\\
1.000000\\
1.000000\\
1.115080\\
1.854581\\
};
\addlegendentry{\Large $\mathcal{E}_{MoR}$}

\addplot[line width = 2.00pt, mycolor5]
table[row sep=\\,y index=0] {
data\\
1.000000\\
1.000000\\
1.000000\\
1.076072\\
1.521895\\
};
\addlegendentry{\Large $\mathcal{E}_{RoS}$}

\addplot[line width = 2.00pt, mycolor4]
table[row sep=\\,y index=0] {
data\\
1.000000\\
1.000000\\
1.000000\\
1.052688\\
1.286692\\
};

\addplot[line width = 2.00pt, mycolor5]
table[row sep=\\,y index=0] {
data\\
1.000000\\
1.000000\\
1.000000\\
1.036283\\
1.212121\\
};
\addplot[line width = 2.00pt, mycolor4]
table[row sep=\\,y index=0] {
data\\
1.000000\\
1.000000\\
1.000000\\
1.039643\\
1.245383\\
};

\addplot[line width = 2.00pt, mycolor5]
table[row sep=\\,y index=0] {
data\\
1.000000\\
1.000000\\
1.001409\\
1.026920\\
1.171352\\
};

\addplot[line width = 2.00pt, mycolor4]
table[row sep=\\,y index=0] {
data\\
1.000000\\
1.000000\\
1.000000\\
1.020709\\
1.111752\\
};

\addplot[line width = 2.00pt, mycolor5]
table[row sep=\\,y index=0] {
data\\
1.000000\\
1.000000\\
1.001614\\
1.015538\\
1.086486\\
};
\addplot[line width = 2.00pt, mycolor4]
table[row sep=\\,y index=0] {
data\\
1.000000\\
1.000000\\
1.000040\\
1.016045\\
1.102951\\
};

\addplot[line width = 2.00pt, mycolor5]
table[row sep=\\,y index=0] {
data\\
1.000000\\
1.000000\\
1.001867\\
1.013407\\
1.079733\\
};

\addplot[line width = 2.00pt, mycolor4]
table[row sep=\\,y index=0] {
data\\
1.000000\\
1.000000\\
1.000000\\
1.012029\\
1.065349\\
};

\addplot[line width = 2.00pt, mycolor5]
table[row sep=\\,y index=0] {
data\\
1.000000\\
1.000000\\
1.002178\\
1.011390\\
1.052391\\
};

  	\end{axis}
  \end{tikzpicture}
 }
\end{minipage}

&
\hspace{-13pt}
\begin{minipage}{0.335\linewidth}
    \resizebox{1.00\textwidth}{!}{
    \begin{tikzpicture}
  	\begin{axis}[
    scale =0.70,
  	legend pos = north east,
   	grid = major,
	xlabel = \LARGE Sample Size $|S|$,
    title = \LARGE (b) $\mathcal{E}_{MoR}$,
  	ymin = 1,
  	ymax = 1.12,
  	xmin = 100,
 	xmax = 500000,
    xmode = log,
  	ticklabel style = {font=\large},
  	ytick={1.00, 1.05, 1.10, 1.15},
   	]

  	\addplot[solid, smooth, line width = 2.00pt, mycolor6] table[x ={sizes}, y = eMoR]{sub_parts/FigData/mean_error_rho_2.dat};
    \addlegendentry{$\rho =0.02$}
  	

    \addplot[solid, smooth, line width = 2.00pt,mycolor5] table[x ={sizes}, y = eMoR]{sub_parts/FigData/mean_error_rho_5.dat};
    \addlegendentry{$\rho =0.05$}


    \addplot[solid, smooth, line width = 2.00pt, mycolor4] table[x ={sizes}, y = eMoR]{sub_parts/FigData/mean_error_rho_10.dat};
    \addlegendentry{$\rho=0.10$}

   \node[] at (axis cs: 100000,1.025) {\textbf{\LARGE $\epsilon = 2\%$}};
   \node[] at (axis cs: 100000,1.055) {\textbf{\LARGE $\epsilon = 5\%$}};
   \addplot[dashed, line width = 1.00pt] (100,1.02) -- (1000000,1.02);
   \addplot[dashed, line width = 1.00pt] (100,1.05) -- (1000000,1.05);
  	\end{axis}
  \end{tikzpicture}
 }
\end{minipage}
&
\hspace{-13pt}
\begin{minipage}{0.335\linewidth}
    \resizebox{1.00\textwidth}{!}{
    \begin{tikzpicture}
  	\begin{axis}[
    scale = 0.70,
  	legend pos = north east,
   	grid = major,
	xlabel = \LARGE Sample Size $|S|$,
    title = \LARGE (c) $\mathcal{E}_{RoS}$,
  	ymin = 1,
  	ymax = 1.12,
  	xmin = 100,
 	xmax = 500000,
    xmode = log,
  	ticklabel style = {font=\large},
  	ytick={1.00, 1.05, 1.10, 1.15},
   	]

      \addplot[solid, smooth, line width = 2.00pt, mycolor6] table[x ={sizes}, y = eRoS]{sub_parts/FigData/mean_error_rho_2.dat};
    \addlegendentry{$\rho =0.02$}


    \addplot[solid, smooth, line width = 2.00pt, mycolor5] table[x ={sizes}, y = eRoS]{sub_parts/FigData/mean_error_rho_5.dat};
    \addlegendentry{$\rho =0.05$}


    \addplot[solid, smooth, line width = 2.00pt, mycolor4] table[x ={sizes}, y = eRoS]{sub_parts/FigData/mean_error_rho_10.dat};
    \addlegendentry{$\rho=0.10$}
    

   \node[] at (axis cs: 100000,1.025) {\textbf{\LARGE $\epsilon = 2\%$}};
   \node[] at (axis cs: 100000,1.055) {\textbf{\LARGE $\epsilon = 5\%$}};
   \addplot[dashed, line width = 1.00pt] (100,1.02) -- (1000000,1.02);
   \addplot[dashed, line width = 1.00pt] (100,1.05) -- (1000000,1.05);
  	\end{axis}
  \end{tikzpicture}
 }
\end{minipage}
\end{tabular}
\end{center}
    \caption{Scale Free ($n = 10^6$). (a) Boxplots of $\error{\method}{}{}$ versus $|S|$ for MoR and RoS. (b) Average $\error{\MoR}{}{}$ versus $|S|$ for different $\ratio{}$. (c) Average $\error{\RoS}{}{}$ versus $|S|$ for different $\ratio{}$.}
    \label{fig_SF_box}
    \Description{}
\end{figure}

\begin{figure*}[t]
\begin{center}
    \definecolor{mycolor1}{rgb}{0.00000,0.44700,0.74100}
\definecolor{mycolor2}{rgb}{0.85000,0.32500,0.09800}
\definecolor{mycolor3}{rgb}{0.92900,0.69400,0.12500}
\definecolor{mycolor4}{rgb}{0.4666,0.67451,0.18824}
\definecolor{mycolor5}{rgb}{0.49412,0.18431,0.30588}
\definecolor{mycolor6}{rgb}{0.30196,0.69020,0.93333}

\begin{tabular}{c c c}
\hspace{-15pt}
\begin{minipage}{0.22\linewidth}
\resizebox{1.00\textwidth}{!}{
 \begin{tikzpicture}
 	\begin{axis}[
    ylabel shift = -5pt,
    scale = 0.70,
 	legend pos = north east,
  	grid = major,
 	ylabel = \LARGE \textrm{$\Pr\left[\mathcal{E}_{\MoR} > 1 + \epsilon\right]$},
        title =  \LARGE \textrm{(a) $\epsilon = 0.05$},
 	ymin = 0,
 	ymax = 1.2,
 	xmin = 100,
	xmax = 500000,
        xmode = log,
 	ticklabel style = {font=\Large},
  	]
    \addplot[solid, smooth, line width = 2.00pt, mycolor6] table[x ={S}, y = MoR]{sub_parts/FigData/Pe_rho_2_epsilon_50.dat};\label{line_SF_20}

    \addplot[solid, smooth, line width = 2.00pt, mycolor5] table[x ={S}, y = MoR]{sub_parts/FigData/Pe_rho_5_epsilon_50.dat};\label{line_SF_30}
        
    \addplot[solid, smooth, line width = 2.00pt, mycolor4] table[x ={S}, y = MoR]{sub_parts/FigData/Pe_rho_10_epsilon_50.dat};\label{line_SF_50}

   \addplot[dashed, dash pattern=on 8pt off 3pt, line width = 2.00pt, mycolor6] table[x ={S}, y = MoRB]{sub_parts/FigData/Pe_rho_2_epsilon_50.dat};

   \addplot[dashed, dash pattern=on 8pt off 3pt, line width = 2.00pt, mycolor5] table[x ={S}, y = MoRB]{sub_parts/FigData/Pe_rho_5_epsilon_50.dat};
   
   \addplot[dashed, dash pattern=on 8pt off 3pt, line width = 2.00pt, mycolor4] table[x ={S}, y = MoRB]{sub_parts/FigData/Pe_rho_10_epsilon_50.dat};

   \addplot[densely dashdotted, line width = 2.00pt, mycolor6] (329768,0.0) -- (329768,1.2);
   \addplot[densely dashdotted, line width = 2.00pt, mycolor5] (131907,0.0) -- (131907,1.2);
   \addplot[densely dashdotted, line width = 2.00pt, mycolor4] (65953,0.0) -- (65953,1.2);

\draw[rotate=0, line width = 1.50pt] (200, 0.35) ellipse (3pt and 10pt);
\node[] at (axis cs: 250,0.50) {\Large \textbf{Sim.}};

   \draw[rotate=0, line width = 1.50pt] (10000, 0.98) ellipse (15pt and 3pt);
   \node[] at (axis cs: 10000,1.09) {\Large  \textbf{Thm \ref{thm:MoR-bound-main}}}; 

   \draw[rotate=0, line width = 1.50pt] (150000, 0.82) ellipse (15pt and 3pt);
   \node[rotate=90] at (axis cs: 45000,0.75) {\Large  \textbf{Thm \ref{thm:sample-size}}};   
 	\end{axis}
    
 \end{tikzpicture}
}
\end{minipage}
     &
\hspace{-10pt}
\begin{minipage}{0.23\linewidth}
\resizebox{1.00\textwidth}{!}{
 \begin{tikzpicture}
 	\begin{axis}[
        ylabel shift = -5pt,
        scale = 0.70,
 	legend pos = north east,
  	grid = major,
	ylabel = \Large \textrm{Sample Size $|S|$},
        title =  \LARGE \textrm{(b)},
 	ymin = 125,
 	ymax = 100000,
 	xmin = 0.065,
	xmax = 0.18,
        ymode = log,
 	ticklabel style = {font=\Large},
    x tick label style={
        /pgf/number format/.cd,
            fixed,
            fixed zerofill,
            precision=2,
        /tikz/.cd
    }]
    \addplot[solid, smooth, line width = 2.00pt,mycolor6] table[x ={epsilon}, y = size]{sub_parts/FigData/size_vs_epsilonMoR_rho_2.dat};
    \addplot[solid, smooth, line width = 2.00pt,mycolor5] table[x ={epsilon}, y = size]{sub_parts/FigData/size_vs_epsilonMoR_rho_5.dat};
    \addplot[solid, smooth, line width = 2.00pt,mycolor4] table[x ={epsilon}, y = size]{sub_parts/FigData/size_vs_epsilonMoR_rho_10.dat};

    \addplot[dashed, dash pattern=on 8pt off 3pt, line width = 2.00pt, mycolor6, smooth] table[x ={epsilon}, y = size]{sub_parts/FigData/size_vs_epsilonMoRB_rho_2.dat};
    \addplot[dashed, dash pattern=on 8pt off 3pt, line width = 2.00pt, mycolor5,smooth] table[x ={epsilon}, y = size]{sub_parts/FigData/size_vs_epsilonMoRB_rho_5.dat};
    \addplot[dashed, dash pattern=on 8pt off 3pt, line width = 2.00pt, mycolor4,smooth] table[x ={epsilon}, y = size]{sub_parts/FigData/size_vs_epsilonMoRB_rho_10.dat};

   \draw[rotate=0, line width = 1.50pt] (0.15, 10000) ellipse (3pt and 12pt);
   \node[] at (axis cs: 0.15, 40000) {\Large \textbf{Thm \ref{thm:MoR-bound-main}}};

   \draw[rotate=0, line width = 1.50pt] (0.12, 800) ellipse (3pt and 13pt);
   \node[] at (axis cs: 0.11,300) {\Large \textbf{Sim.}};
 	\end{axis}

 \end{tikzpicture}
}
\end{minipage}
     &
\hspace{-10pt}
\begin{minipage}{0.22\linewidth}
\resizebox{1.00\textwidth}{!}{
 \begin{tikzpicture}
 	\begin{axis}[
        legend pos = south west,
        scale = 0.70,
  	grid = major,
        title =  \LARGE \textrm{(c) $\epsilon = 0.05$},
 	ymin = 200,
 	ymax = 100000,
 	xmin = 0.05,
	xmax = 0.20,
        ymode = log,
 	ticklabel style = {font=\Large},
  	]
  	\addplot[solid, smooth, line width = 2.00pt,mycolor5] table[x ={rho}, y = size]{sub_parts/FigData/size_vs_rhoMoR.dat};
    \addlegendentry{\large Sim.}
    \addplot[dashed, dash pattern=on 8pt off 3pt,smooth,line width = 2.00pt,mycolor5] table[x ={rho}, y = size]{sub_parts/FigData/size_vs_rhoMoRB.dat};
    \addlegendentry{\large Thm \ref{thm:MoR-bound-main}}
   
 	\end{axis}
 \end{tikzpicture}
}
\end{minipage}
     \\
\hspace{-15pt}
\begin{minipage}{0.22\linewidth}
\resizebox{1.00\textwidth}{!}{
 \begin{tikzpicture}
 	\begin{axis}[
        ylabel shift = -5pt,
        scale = 0.70,
 	legend pos = north east,
  	grid = major,
 	ylabel = \LARGE \textrm{$\Pr\left[\mathcal{E}_{\RoS} > 1 + \epsilon\right]$},
	xlabel = \LARGE \textrm{Sample Size $|S|$},
        title =  \LARGE \textrm{(d) $\epsilon = 0.05$},
 	ymin = 0,
 	ymax = 1.2,
 	xmin = 100,
	xmax = 500000,
        xmode = log,
 	ticklabel style = {font=\Large},
  	]

    \addplot[solid, smooth, line width = 2.00pt, mycolor6] table[x ={S}, y = RoS]{sub_parts/FigData/Pe_rho_2_epsilon_50.dat};
    \addplot[solid, smooth, line width = 2.00pt, mycolor4] table[x ={S}, y = RoS]{sub_parts/FigData/Pe_rho_10_epsilon_50.dat};
    \addplot[dashed, dash pattern=on 8pt off 3pt, line width = 2.00pt, mycolor6] table[x ={S}, y = RoSB]{sub_parts/FigData/Pe_rho_2_epsilon_50.dat};

    \addplot[dashed, dash pattern=on 8pt off 3pt, line width = 2.00pt, mycolor5] table[x ={S}, y = RoSB]{sub_parts/FigData/Pe_rho_5_epsilon_50.dat};
   
    \addplot[dashed, dash pattern=on 8pt off 3pt, line width = 2.00pt, mycolor4] table[x ={S}, y = RoSB]{sub_parts/FigData/Pe_rho_10_epsilon_50.dat};
    
    \addplot[densely dashdotted, line width = 2.00pt, mycolor6] (329768,0.0) -- (329768,1.2);
    \addplot[densely dashdotted, line width = 2.00pt, mycolor5] (131907,0.0) -- (131907,1.2);
    \addplot[densely dashdotted, line width = 2.00pt, mycolor4] (65953,0.0) -- (65953,1.2);

\draw[rotate=0, line width = 1.50pt] (200, 0.25) ellipse (3pt and 12pt);
\node[] at (axis cs: 250,0.45) {\Large \textbf{Sim.}};

\draw[rotate=0, line width = 1.50pt] (4000, 1.00) ellipse (15pt and 3pt);
\node[] at (axis cs: 4000,1.10) {\Large \textbf{Thm \ref{thm:RoS-bound}}};

   \draw[rotate=0, line width = 1.50pt] (150000, 0.82) ellipse (15pt and 3pt);
   \node[rotate=90] at (axis cs: 45000,0.75) {\Large \textbf{Thm \ref{thm:sample-size}}};
   
 	\end{axis}    
 \end{tikzpicture}
}
\end{minipage}
&
\hspace{-10pt}
\begin{minipage}{0.23\linewidth}
\resizebox{1.00\textwidth}{!}{
 \begin{tikzpicture}
 	\begin{axis}[
        ylabel shift = -5pt,
        scale = 0.70,
 	legend pos = north east,
  	grid = major,
 	xlabel = \huge $\epsilon$,
	ylabel = \Large \textrm{Sample Size $|S|$},
        title =  \LARGE \textrm{(e)},
 	ymin = 150,
 	ymax = 100000,
 	xmin = 0.065,
	xmax = 0.18,
        ymode = log,
 	ticklabel style = {font=\Large},
        x tick label style={
        /pgf/number format/.cd,
            fixed,
            fixed zerofill,
            precision=2,
        /tikz/.cd
    }
  	]
    \addplot[solid, smooth, line width = 2.00pt, mycolor6] table[x ={epsilon}, y = size]{sub_parts/FigData/size_vs_epsilonRoS_rho_2.dat};
    \addplot[solid, smooth, line width = 2.00pt, mycolor5] table[x ={epsilon}, y = size]{sub_parts/FigData/size_vs_epsilonRoS_rho_5.dat};
    \addplot[solid, smooth, line width = 2.00pt, mycolor4] table[x ={epsilon}, y = size]{sub_parts/FigData/size_vs_epsilonRoS_rho_10.dat};

    \addplot[dashed, dash pattern=on 8pt off 3pt, line width = 2.00pt, mycolor6, smooth] table[x ={epsilon}, y = size]{sub_parts/FigData/size_vs_epsilonRoSB_rho_2.dat};
    \addplot[dashed, dash pattern=on 8pt off 3pt, line width = 2.00pt, mycolor5, smooth] table[x ={epsilon}, y = size]{sub_parts/FigData/size_vs_epsilonRoSB_rho_5.dat};
    \addplot[dashed, dash pattern=on 8pt off 3pt, line width = 2.00pt, mycolor4, smooth] table[x ={epsilon}, y = size]{sub_parts/FigData/size_vs_epsilonRoSB_rho_10.dat};

   \draw[rotate=0, line width = 1.50pt] (0.15, 3500) ellipse (2pt and 12pt);
   \node[] at (axis cs: 0.15,12000) {\Large \textbf{Thm \ref{thm:RoS-bound}}};

   \draw[rotate=0, line width = 1.50pt] (0.08, 800) ellipse (3pt and 14pt);
   \node[] at (axis cs: 0.092,2200) {\Large \textbf{Sim.}};
 	\end{axis}

 \end{tikzpicture}
}
\end{minipage}
&
\hspace{-10pt}
\begin{minipage}{0.22\linewidth}
\resizebox{1.00\textwidth}{!}{
 \begin{tikzpicture}
 	\begin{axis}[
        legend pos = north east,
        scale = 0.70,
        grid = major,
 	xlabel = \huge $\rho$,
        title =  \LARGE \textrm{(f) $\epsilon = 0.05$},
 	ymin = 200,
 	ymax = 100000,
 	xmin = 0.05,
	xmax = 0.20,
        ymode = log,
 	ticklabel style = {font=\Large},
  	]
    \addplot[solid, smooth, line width = 2.00pt, mycolor5] table[x ={rho}, y = size]{sub_parts/FigData/size_vs_rhoRoS.dat};
    \addlegendentry{\large Sim.}
    \addplot[dashed, dash pattern=on 8pt off 3pt, line width = 2.00pt,mycolor5] table[x ={rho}, y = size]{sub_parts/FigData/size_vs_rhoRoSB.dat};
    \addlegendentry{\large Thm \ref{thm:RoS-bound}}

 	\end{axis}
    
 \end{tikzpicture}
}
\end{minipage}
\end{tabular}
\end{center}
    \caption{Scale Free ($n = 10^6$). Simulation curves and bounds for (\ref{line_SF_20}) $\ratio{} = 0.02$, (\ref{line_SF_30}) $\ratio{} = 0.05$, and (\ref{line_SF_50}) $\ratio{} = 0.10$. (a) $\Pr[\error{\MoR}{}{} > 1 + \epsilon]$ versus $|S|$ for $\epsilon = 0.05$. (b) Sample size versus $\epsilon$. (c) Sample size versus $\ratio{}$ for $\epsilon = 0.05$. (d) $\Pr[\error{\RoS}{}{} > 1 + \epsilon]$ versus $|S|$ for $\epsilon = 0.05$. (e) Sample size versus $\epsilon$. (f) Sample size versus $\ratio{}$ for $\epsilon = 0.05$.}
    \label{fig_SF_bounds}
    \Description{}
\end{figure*}

\subsection{Scale-Free}

In this section, we evaluate the accuracy of the 
estimators using scale-free graphs with size $n=10^6$.

\noindent
\textbf{Error versus sample size.} Fig.~\ref{fig_SF_box}(a) shows the boxplots of $\error{\MoR}{}{}$ and $\error{\RoS}{}{}$ versus $|S|$ for $\ratio{} = 0.05~(5\%)$. Errors decrease with an increase in $|S|$. Additionally, the $\RoS$ estimator yields lower errors than the $\MoR$ estimator across various sample sizes. Fig.~\ref{fig_SF_box}(b) and (c) illustrate the average $\error{\MoR}{}{}$ and $\error{\RoS}{}{}$, respectively, versus $|S|$ for $\ratio{} = 0.02$, $0.05$, and $0.10$. 
The results show that higher \ratio{} leads to lower estimation errors, indicating improved estimator performance. Moreover, the $\estimate{\RoS}{} $ estimator demonstrates better accuracy than the $\estimate{\MoR}{} $ estimator as the sample size increases. These findings suggest that larger sample sizes and higher hidden population \prevalence lead to more accurate estimations, with the $\RoS$ estimator providing superior performance.


\noindent
\textbf{Analysis of $\error{\MoR}{}{}$.} In this case, we analyze the behavior of $\error{\MoR}{}{}$ in the Scale-Free model under various conditions. Fig.~\ref{fig_SF_bounds}(a) shows the probability $\Pr[\error{\MoR}{}{} > 1 + \epsilon]$ versus $|S|$ for $\epsilon = 0.05$ and different actual $\ratio{} = 0.02$, $0.05$, and $0.10$. This figure includes the theoretical bounds from Thms \ref{thm:MoR-bound-main} and \ref{thm:sample-size}, with Thm.~\ref{thm:sample-size} bounds obtained using $\alpha = 1/2$. The bounds from Thm.~\ref{thm:sample-size} are close to the simulation results. Fig.~\ref{fig_SF_bounds}(b) displays the sample size required to achieve $\Pr[\error{\MoR}{}{} > 1 + \epsilon] = 0.05$ versus $\epsilon$ for different $\ratio{}$. The sample sizes required from Thm.~\ref{thm:MoR-bound-main} are approximately one order of magnitude higher than the simulations. Fig.~\ref{fig_SF_bounds}(c) shows the required sample size versus the actual $\ratio{}$ for $\Pr[\error{\MoR}{}{} > 1 + \epsilon] = 0.05$. 

\noindent
\textbf{Analysis of $\error{\RoS}{}{}$.} We evaluate the behavior of $\error{\RoS}{}{}$ using the Scale-Free network model, employing the bounds derived in 
Thm.~\ref{thm:sample-size} and \ref{thm:RoS-bound}. Fig.~\ref{fig_SF_bounds}(d) shows the probability $\Pr[\error{\RoS}{}{} > 1 + \epsilon]$ versus sample size $|S|$ for $\epsilon = 0.05$ and $\ratio{} = 0.02$, $0.05$, and $0.10$. The bounds from Thm.~\ref{thm:sample-size} and \ref{thm:RoS-bound} are compared to the simulation results. Fig.~\ref{fig_SF_bounds}(e) illustrates the required sample size to achieve $\Pr[\error{\RoS}{}{} > 1 + \epsilon] = 0.05$ versus $\epsilon$ for $\ratio{} = 0.05$. The results indicate that the sample sizes predicted by Thm.~\ref{thm:RoS-bound} are slightly higher than those obtained from simulations. Fig.~\ref{fig_SF_bounds}(f) presents the required sample size versus \ratio{} for $\Pr[\error{\RoS}{}{} > 1 + \epsilon] = 0.05$ and $\epsilon = 0.05$.

\begin{table}
\caption{Summary of characteristics of friendship networks from Deezer dataset \cite{rozemberczki2019gemsec}.}
\scriptsize
    \centering
    \begin{sc}
    \resizebox{0.5\linewidth}{!}{
    \begin{tabular}{| c | c |c c c |}
    \hline
    
         \multicolumn{2}{|c|}{} & \multicolumn{3}{c|}{Network}  \\
         \cline{3-5}
         \multicolumn{2}{|c|}{} & Croatia & Hungary & Romania \\ 
         \hline
        \multicolumn{2}{|c|}{Nodes} & $54,573$ & $41,538$ & $41,773$ \\
         \multicolumn{2}{|c|}{Edges} & $498,202$ & 
         $222,887$ & $125,826$ \\
         \multicolumn{2}{|c|}{Avg Node Deg} & $18.26$ & $9.38$ & $6.02$ \\
         \hline
    \end{tabular}
    }
    \end{sc}
    \label{tab:deezer-data}
\end{table}

\begin{figure}[t!]
\begin{center}
    \definecolor{mycolor1}{rgb}{0.00000,0.44700,0.74100}
\definecolor{mycolor2}{rgb}{0.85000,0.32500,0.09800}
\definecolor{mycolor3}{rgb}{0.92900,0.69400,0.12500}
\definecolor{mycolor4}{rgb}{0.46667,0.67451,0.18824}
\definecolor{mycolor5}{rgb}{0.49412,0.18431,0.30588}
\definecolor{mycolor6}{rgb}{0.30196,0.69020,0.93333}

\begin{tabular}{c c c}
    \hspace{-10pt}
    \begin{minipage}{.3\linewidth}
    \resizebox{1.00\textwidth}{!}{
    \begin{tikzpicture}
        \begin{axis}[
            scale only axis,
            width=\textwidth,
            ylabel style = {yshift=-0.1cm},
            xlabel = \small Degree,
            xlabel style = {yshift=0.1cm},
            title = \small Croatia,
            ymin=0,
            ymax=0.1,
            xmin=0,
            xmax=50,
            ticklabel style = {font=\scriptsize},
            ytick ={0.1},
        ]
        \addplot[ybar,bar width=1.00,fill=mycolor1!80,opacity=0.5] table[x ={x},y = y]{sub_parts/FigData/data_hist_HR.dat};
\end{axis}
\end{tikzpicture}
    }
    \end{minipage}
    &
\hspace{-12pt}
    \begin{minipage}{.3\linewidth}
    \resizebox{1.00\textwidth}{!}{
    \begin{tikzpicture}
        \begin{axis}[
            scale only axis,
            width=\textwidth,
            ylabel style = {yshift=-0.1cm},
            xlabel = \small Degree,
            xlabel style = {yshift=0.1cm},
            title = \small Hungary,
            ymin=0,
            ymax=0.1,
            xmin=0,
            xmax=50,
            ticklabel style = {font=\scriptsize},
            ytick ={0.1},
        ]
        \addplot[ybar,bar width=1.00,fill=mycolor1!80,opacity=0.5] table[x ={x},y = y]{sub_parts/FigData/data_hist_HU.dat};
\end{axis}
\end{tikzpicture}
    }
    \end{minipage}
&
    \hspace{-12pt}
    \begin{minipage}{.3\linewidth}
    \resizebox{1.00\textwidth}{!}{
    \begin{tikzpicture}
        \begin{axis}[
            scale only axis,
            width=\textwidth,
            ylabel style = {yshift=-0.1cm},
            xlabel = \small Degree,
            xlabel style = {yshift=0.1cm},
            title = \small Romania,
            ymin=0,
            ymax=0.2,
            xmin=0,
            xmax=50,
            ticklabel style = {font=\scriptsize},
            ytick ={0.1,0.2},
        ]
        \addplot[ybar,bar width=1.00,fill=mycolor1!80,opacity=0.5] table[x ={x},y = y]{sub_parts/FigData/data_hist_RO.dat};
\end{axis}
\end{tikzpicture}
    }
    \end{minipage}
\end{tabular}\vspace{-10pt}
\end{center}
    \caption{Deezer Music Dataset. Histograms of the node degree for the networks obtained from Croatia, Hungary, and Romania.}
    \Description{}
    \label{fig:hist1}
\end{figure}

\begin{figure}
\begin{center}
    \input{sub_parts/FigTex/Bound_MoR_StanfordGraphs}
\end{center}
    \caption{Gemsec Deezer dataset. Hungary. ($n = 47,538$). $\Pr\left[\mathcal{E}_{\mathcal{M}} > 1 + \epsilon\right]$ versus $|S|$ obtained using (\ref{crv:simMoR}) MoR, (\ref{crv:simRoS}) RoS,(\ref{crv:MoRbndmain}) Thm~\ref{thm:MoR-bound-main}, (\ref{crv:bndRoSProb}) Thm.~\ref{thm:RoS-bound-main}, and (\ref{crv:smpsz}) Thm.~\ref{thm:sample-size} for $\epsilon = 0.10$ and music genres: (a) Soundtracks, (b) Comedy, (c) Singer \& Songwriter, (d) Film-Games, (e) Alternative, and (f) Dance.}
    \label{fig_realdata_boundsMoR}
    \Description{}
\end{figure}


 \begin{figure*}[ht]
 \begin{center}
     \input{sub_parts/FigTex/Bound_MoR_HR}
 \end{center}
     \caption{Gemsec Deezer dataset. Croatia. ($n = 54,573$). $\Pr\left[\mathcal{E}_{\mathcal{M}} > 1 + \epsilon\right]$ versus $|S|$ obtained using (\ref{crv:simMoR}) MoR, (\ref{crv:simRoS}) RoS,(\ref{crv:MoRbndmain}) Thm~\ref{thm:MoR-bound-main}, (\ref{crv:bndRoSProb}) Thm.~\ref{thm:RoS-bound-main}, and (\ref{crv:smpsz}) Thm.~\ref{thm:sample-size} for $\epsilon = 0.10$ and music genres: (a) Trance, (b) Kids, (c) Reggae, (d) Latin Music, (e) Singer \& Songwriter, and (f) Alternative.}
     \Description{}
     \label{fig_realdata_boundsMoR_HR}
 \end{figure*}

 \begin{figure*}[ht]
 \begin{center}
     \input{sub_parts/FigTex/Bound_MoR_RO}
 \end{center}
     \caption{Gemsec Deezer dataset. Romania. ($n = 41,773$). $\Pr\left[\mathcal{E}_{\mathcal{M}} > 1 + \epsilon\right]$ versus $|S|$ obtained using (\ref{crv:simMoR}) MoR, (\ref{crv:simRoS}) RoS,(\ref{crv:MoRbndmain}) Thm~\ref{thm:MoR-bound-main}, (\ref{crv:bndRoSProb}) Thm.~\ref{thm:RoS-bound-main}, and (\ref{crv:smpsz}) Thm.~\ref{thm:sample-size} for $\epsilon = 0.10$ and music genres: (a) Trance, (b) Hard Rock, (c) Indie Pop-Folk, (d) R \& B, (e) Electro, and (f) Dance.}
     \Description{}
     \label{fig_realdata_boundsMoR_RO}
 \end{figure*}

 \begin{figure*}[ht]
     \centering
     \includegraphics[width=\linewidth]{sub_parts/EpsFigs/HU_hist_Rs.png}
     \caption{Gamesec Deezer dataset. Hungary. ($n=47,538$). Histograms of $R_s$ for different sample sizes. Each histogram is obtained from $10,000$ realizations of the sampling process.}
     \Description{}
     \label{fig:histRsHU}
 \end{figure*}

 \begin{figure*}[ht]
     \centering
     \includegraphics[width=\linewidth]{sub_parts/EpsFigs/HR_hist_Rs.png}
     \caption{Gamesec Deezer dataset. Croatia. ($n=54,573$). Histograms of $R_s$ for different sample sizes. Each histogram is obtained from $10,000$ realizations of the sampling process.}
     \Description{}
     \label{fig:histRsHR}
 \end{figure*}

 \begin{figure*}[ht]
     \centering
     \includegraphics[width=\linewidth]{sub_parts/EpsFigs/RO_hist_Rs.png}
     \caption{Gamesec Deezer dataset. Romania. ($n=41,773$). Histograms of $R_s$ for different sample sizes. Each histogram is obtained from $10,000$ realizations of the sampling process.}
     \Description{}
     \label{fig:histRsRO}
 \end{figure*}

\section{Results with Real Networks}
\label{sec:simulation}

\noindent
\textbf{Dataset.} This section evaluates the performance of the hidden population rate estimators by using (undirected) friendship networks extracted from Deezer, a popular music streaming platform, in November 2017. Specifically, these datasets contain friendship networks of the Deezer users from three European countries: Hungary, Romania, and Croatia. In these datasets, each node represents a Deezer user, and each edge denotes a mutual friendship between users. Moreover, based on their liked song lists, each node is labeled with a set of music genres (from a list of 84 music genres) they have shown interest in. These datasets were collected by Rozemberczki et al. as a part of their study on graph embedding techniques \cite{rozemberczki2019gemsec}. Additionally, these datasets are publicly available through the Stanford Large Network Dataset Collection: \url{https://snap.stanford.edu/data/gemsec-Deezer.html}. 



Table \ref{tab:deezer-data} summarizes the characteristics of the friendship networks extracted from the Gemsec Deezer dataset. These networks comprise roughly between $41,000$ and $55,000$ nodes, and between $125,000$ and $500,000$ edges. Notice that the networks exhibit different averages of node degree. Fig.~\ref{fig:hist1} illustrates the node degree distributions for the friendship networks obtained from the three countries. 

\noindent
\textbf{Analysis of $\error{\MoR}{}{}$.} Our study identified user groups interested in particular music genres. For example, we have chosen user sets with preferences for soundtracks, comedy, singer \& songwriter, film games, alternative, and dance as hidden populations in Hungary. We selected the music genres for each country to encompass a range of hidden population rates from 1\% to 50\%. Our approach involves utilizing varying sample sizes for each hidden population rate, and for each sample size, we generate 10,000 realizations of the sampling process. We calculate the hidden population rate at each trial using the $\MoR$ estimator and the corresponding $\mathcal{E}_{\MoR}$. Fig.~\ref{fig_realdata_boundsMoR} shows the probabilities $\Pr\left[\mathcal{E}_{\MoR} > 1 + \epsilon\right]$ as a function of the sample size $|S|$ obtained through simulations, Thm.~\ref{thm:MoR-bound-main}, and Thm.~\ref{thm:sample-size} for $\epsilon = 10\%$ and the different music genres selected from the Hungary dataset. Notice that the sample size bound in Thm.~\ref{thm:sample-size} is calculated using $\alpha=1/2$. 
Fig.~\ref{fig_realdata_boundsMoR_HR} and \ref{fig_realdata_boundsMoR_RO} present $\Pr\left[\mathcal{E}_{\MoR} > 1 + \epsilon\right]$ versus $|S|$ for $\epsilon = 10\%$ and the music genres selected from the Croatia and Romania datasets, respectively. 
As can be seen in these figures, bounds from Thm.~\ref{thm:MoR-bound-main} and Thm.~\ref{thm:sample-size} are conservative concerning probabilities obtained from simulations.

\begin{table*}[ht]
\caption{Summary of Analytical Error Bounds for NSUM Hidden Population Rate Estimators. Sample $S \subseteq V$ is random uniform with $m=|S|$, $\Reach{S}=\sum_{v \in S} \reach{v}$, and any $\beta=1+\epsilon > 1$, any $\alpha>0$, and any $\delta \in (0,1)$.}
\tiny
    \centering
    \begin{sc}
    \begin{tabular}{|c|c|c|c|}
    \hline
    \hline
     Case  & Description & Bound & Theorem \\
    \hline
    \hline
    Adversarial & Lower bound, & \multirow{2}{*}{$\error{\method}{}{} \geq \sqrt{(n-1)/2}$ } & \multirow{2}{*}{\ref{t-lower-any}} \\
    Instances & Deterministic $\method$ & & \\
    \hline
    \hline
    \multirow{8}{*}{\begin{tabular}{c}
     Random \\
     Networks
    \end{tabular}} & Upper Bound \MoR & $\Pr\brac{\error{\MoR}{}{} > \beta } \leq \paren{\frac{e^{\beta-1}}{\beta^{\beta}}}^{m \ratio{}} + \paren{\frac{e^{\frac{1}{\beta}-1}}{\beta^{-1/\beta}}}^{m \ratio{}}$ & \ref{thm:MoR-bound-main} \\
    \cline{2-4} & Upper Bound \RoS & $\Pr\brac{\error{\RoS}{}{} > \beta } \leq \sum_{\Reach{}} \paren{ \paren{\frac{e^{\beta-1}}{\beta^{\beta}}}^{\Reach{} \ratio{}} + \paren{\frac{e^{\frac{1}{\beta}-1}}{\beta^{-1/\beta}}}^{\Reach{} \ratio{}} } \Pr\Brack{\Reach{S}=\Reach{}}$ & \ref{thm:RoS-bound-main} \\
    \cline{2-4}
    & \shortstack{Sample size \\ $\method \in \{\MoR, \RoS\}$ 
    } 
    & $m = |S| \geq \frac{\ln{2} + \alpha \ln{n}}{\ratio{} (1 -\frac{1}{\beta} (\ln{\beta} + 1))} \implies \Pr\brac{\error{\method}{}{} > \beta } \leq 1/n^\alpha$
    &  \ref{thm:sample-size} \\
    \hline
    \hline
    \multirow{5}{*}{\begin{tabular}{c}
     Erdős–Rényi \\
     Networks
    \end{tabular}} & Upper Bound \MoR & See above & \ref{thm:MoR-erdos-renyi} \\
    \cline{2-4} & \shortstack{Upper Bound \RoS \\ $\mu = m p(n-1)$} & 
    $\Pr\brac{\error{\RoS}{}{} > \beta } \leq \paren{\frac{e^{-\delta}}{(1-\delta)^{(1-\delta)}}}^{\mu} + 
  \paren{\frac{e^{\beta-1}}{\beta^{\beta}}}^{(1-\delta) \mu \ratio{}} + \paren{\frac{e^{\frac{1}{\beta}-1}}{\beta^{-1/\beta}}}^{(1-\delta) \mu \ratio{}}$
    & \ref{thm:RoS-erdos-renyi} \\
    \hline
    \hline
    \multirow{5}{*}{\begin{tabular}{c}
     Scale-Free \\
     Networks
    \end{tabular}}& Upper Bound \MoR & See above & \ref{thm:MoR-bound-main} \\
    \cline{2-4}
     & \shortstack{Upper Bound \RoS, 
     $\gamma>2$\\ $\mu \approx m \frac{1-\gamma}{2-\gamma} \frac{1-(n-1)^{2-\gamma}}{1-(n-1)^{1-\gamma}}$} & 
    $\Pr\brac{\error{\RoS}{}{} > \beta } \leq \paren{\frac{e^{-\delta}}{(1-\delta)^{(1-\delta)}}}^{\mu} + 
  \paren{\frac{e^{\beta-1}}{\beta^{\beta}}}^{(1-\delta) \mu \ratio{}} + \paren{\frac{e^{\frac{1}{\beta}-1}}{\beta^{-1/\beta}}}^{(1-\delta) \mu \ratio{}}$
    & \ref{thm:RoS-bound} \\
    \hline
    \hline
    \end{tabular}
    \end{sc}
    \label{tab:summary}
\end{table*}

\noindent
\textbf{Analysis of $\error{\RoS}{}{}$.} We have utilized the $\RoS$ estimator to calculate hidden population rates and determined the corresponding $\mathcal{E}_{\MoR}$. In Fig.~\ref{fig_realdata_boundsMoR} , we have also depicted the probabilities $\Pr\left[\mathcal{E}_{\RoS} > 1 + \epsilon\right]$ against $|S|$ obtained from simulations and Thm.~\ref{thm:RoS-bound-main} for $\epsilon=0.10$ and the music genres selected from the Hungary dataset. To derive the bound as defined in \Cref{thm:RoS-bound-main}, we have used the probability mass distribution of the variable $R_S$ based on 10,000 realizations of the sampling process for each sample size. 
Figs \ref{fig:histRsHU}, \ref{fig:histRsHR}, and \ref{fig:histRsRO} display the histograms of $R_S$ for different countries and sample sizes. 
~\ref{fig_realdata_boundsMoR_HR} and \ref{fig_realdata_boundsMoR_RO} also illustrate $\Pr\left[\mathcal{E}_{\RoS} > 1 + \epsilon\right]$ versus $|S|$ for $\epsilon = 10\%$ and genres selected from the Croatia and Romania datasets, respectively. 
In this case, the bounds derived from Thm.~\ref{thm:RoS-bound-main} are slightly higher than those obtained from simulations.

\section{Conclusions and Future Work}

In this paper, we derived analytical error bounds for NSUM hidden population rate estimators. The summary of the derived analytical bounds is presented in Table \ref{tab:summary}. Notice that these analytical bounds were determined assuming that sampled nodes accurately report their in-degree and the number of neighboring nodes that belong to the hidden population. Specifically, in adversarial scenarios, the error can remain high even when ARD of all nodes is available.  For such cases, we establish a lower bound. In addition, we have derived analytical error upper bounds for two NSUM hidden population rate estimators (the mean of rates, \MoR, and the rate of sums, \RoS) in random networks and provided the analytical error bounds for Erdős–Rényi and Scale-Free networks. 

On the other hand, extensive numerical simulations have been conducted on synthetic and real networks to evaluate the behavior of the error bounds as the sample size increases. Simulations have allowed us to estimate the minimum sample size required to achieve a given error probability for different error thresholds and hidden population rates. Particularly, the results of the numerical experiments carried out support and corroborate the theoretical analysis on the error bounds.

To our knowledge, this paper is the first to present theoretical analytical bounds on the performance of NSUM methods, opening many avenues for future research. For instance, it will be of interest to study which additional information can be collected from the samples in order to overcome the worst-case results. Regarding random networks, topologies beyond Erdős–Rényi and Scale-Free can be explored. For instance, the Stochastic Block Model or Hyperbolic topologies seem of special interest. It is also worth analyzing the worst error that may arise when the topology is random but the dissemination of the hidden sub-population is controlled by an adversary after the network has been defined. Finally, efficient methods to obtain information about the social network from ARD beyond the size of the hidden sub-population, e.g., the degree distribution or the topology, would be of great interest.

\section{Acknowledgments}

This paper has been funded by project PID2022-140560OB-I00 (DRONAC)
funded by MICIU/AEI /10.13039/501100011033 and ERDF, EU. This research is part of the I+D+i projects PID2022-137243OB-I00 funded by MCIN/AEI/10.13039/501100011033 and European Union NextGenerationEU/PRTR and the project CuidaNSUM of the Instituto de las Mujeres. This initiative has also been partially carried out within the framework of the Recovery, Transformation and Resilience Plan funds, financed by the European Union (Next Generation) through the grant ANTICIPA (INCIBE) and the ENIA 2022 Chairs for the creation of university-industry chairs in AI-AImpulsa: UC3M-Universia. The work of Sergio Díaz-Aranda has been funded by \textit{Comunidad de Madrid} predoctoral grant PIPF-2022/COM-24467.

\bibliography{main_nsum-performance-bound}


\begin{thebibliography}{26}


\ifx \showCODEN    \undefined \def \showCODEN     #1{\unskip}     \fi
\ifx \showISBNx    \undefined \def \showISBNx     #1{\unskip}     \fi
\ifx \showISBNxiii \undefined \def \showISBNxiii  #1{\unskip}     \fi
\ifx \showISSN     \undefined \def \showISSN      #1{\unskip}     \fi
\ifx \showLCCN     \undefined \def \showLCCN      #1{\unskip}     \fi
\ifx \shownote     \undefined \def \shownote      #1{#1}          \fi
\ifx \showarticletitle \undefined \def \showarticletitle #1{#1}   \fi
\ifx \showURL      \undefined \def \showURL       {\relax}        \fi
\providecommand\bibfield[2]{#2}
\providecommand\bibinfo[2]{#2}
\providecommand\natexlab[1]{#1}
\providecommand\showeprint[2][]{arXiv:#2}

\bibitem[Ahmadi-Gohari et~al\mbox{.}(2019)]%
        {ahmadi2019twelve}
\bibfield{author}{\bibinfo{person}{Milad Ahmadi-Gohari}, \bibinfo{person}{Farzaneh Zolala}, \bibinfo{person}{Abedin Iranpour}, {and} \bibinfo{person}{Mohammad~Reza Baneshi}.} \bibinfo{year}{2019}\natexlab{}.
\newblock \showarticletitle{Twelve-hour before driving prevalence of alcohol and drug use among heavy vehicle drivers in south east of Iran using network scale up}.
\newblock \bibinfo{journal}{\emph{Addiction \& health}} \bibinfo{volume}{11}, \bibinfo{number}{4} (\bibinfo{year}{2019}), \bibinfo{pages}{256}.
\newblock


\bibitem[Bernard et~al\mbox{.}(1988)]%
        {bernard1988many}
\bibfield{author}{\bibinfo{person}{H Bernard}, \bibinfo{person}{E Johnsen}, \bibinfo{person}{P Killworth}, {and} \bibinfo{person}{S Robinson}.} \bibinfo{year}{1988}\natexlab{}.
\newblock \bibinfo{booktitle}{\emph{{How many people died in the Mexico City earthquake}}}.
\newblock \bibinfo{type}{{T}echnical {R}eport}. \bibinfo{institution}{University of Florida}.
\newblock


\bibitem[Bernard et~al\mbox{.}(1991)]%
        {bernard1991estimating}
\bibfield{author}{\bibinfo{person}{H~Russell Bernard}, \bibinfo{person}{Eugene~C Johnsen}, \bibinfo{person}{Peter~D Killworth}, {and} \bibinfo{person}{Scott Robinson}.} \bibinfo{year}{1991}\natexlab{}.
\newblock \showarticletitle{Estimating the size of an average personal network and of an event subpopulation: Some empirical results}.
\newblock \bibinfo{journal}{\emph{Social Science Research}} \bibinfo{volume}{20}, \bibinfo{number}{2} (\bibinfo{year}{1991}), \bibinfo{pages}{109--121}.
\newblock


\bibitem[Bernard et~al\mbox{.}(2001)]%
        {bernard2001estimating}
\bibfield{author}{\bibinfo{person}{H~Russell Bernard}, \bibinfo{person}{Peter~D Killworth}, \bibinfo{person}{Eugene~C Johnsen}, \bibinfo{person}{Gene~A Shelley}, {and} \bibinfo{person}{Christopher McCarty}.} \bibinfo{year}{2001}\natexlab{}.
\newblock \showarticletitle{Estimating the ripple effect of a disaster}.
\newblock \bibinfo{journal}{\emph{Connections}} \bibinfo{volume}{24}, \bibinfo{number}{2} (\bibinfo{year}{2001}), \bibinfo{pages}{18--22}.
\newblock


\bibitem[Bernstein et~al\mbox{.}(2013)]%
        {DBLP:conf/chi/BernsteinBBK13}
\bibfield{author}{\bibinfo{person}{Michael~S. Bernstein}, \bibinfo{person}{Eytan Bakshy}, \bibinfo{person}{Moira Burke}, {and} \bibinfo{person}{Brian Karrer}.} \bibinfo{year}{2013}\natexlab{}.
\newblock \showarticletitle{Quantifying the invisible audience in social networks}. In \bibinfo{booktitle}{\emph{Proceedings of the SIGCHI Conference on Human Factors in Computing Systems}} (Paris, France) \emph{(\bibinfo{series}{CHI '13})}. \bibinfo{publisher}{Association for Computing Machinery}, \bibinfo{address}{New York, NY, USA}, \bibinfo{pages}{21–30}.
\newblock
\href{https://doi.org/10.1145/2470654.2470658}{doi:\nolinkurl{10.1145/2470654.2470658}}


\bibitem[Chen et~al\mbox{.}(2016)]%
        {neurips-chen-2016}
\bibfield{author}{\bibinfo{person}{Lin Chen}, \bibinfo{person}{Amin Karbasi}, {and} \bibinfo{person}{Forrest~W. Crawford}.} \bibinfo{year}{2016}\natexlab{}.
\newblock \showarticletitle{Estimating the Size of a Large Network and its Communities from a Random Sample}. In \bibinfo{booktitle}{\emph{Advances in Neural Information Processing Systems}}, \bibfield{editor}{\bibinfo{person}{D.~Lee}, \bibinfo{person}{M.~Sugiyama}, \bibinfo{person}{U.~Luxburg}, \bibinfo{person}{I.~Guyon}, {and} \bibinfo{person}{R.~Garnett}} (Eds.), Vol.~\bibinfo{volume}{29}. \bibinfo{publisher}{Curran Associates, Inc.}, \bibinfo{address}{Barcelona, Spain}, \bibinfo{pages}{3072--3080}.
\newblock


\bibitem[Clauset et~al\mbox{.}(2009)]%
        {clauset2009power}
\bibfield{author}{\bibinfo{person}{Aaron Clauset}, \bibinfo{person}{Cosma~Rohilla Shalizi}, {and} \bibinfo{person}{Mark~EJ Newman}.} \bibinfo{year}{2009}\natexlab{}.
\newblock \showarticletitle{Power-law distributions in empirical data}.
\newblock \bibinfo{journal}{\emph{SIAM review}} \bibinfo{volume}{51}, \bibinfo{number}{4} (\bibinfo{year}{2009}), \bibinfo{pages}{661--703}.
\newblock


\bibitem[Crawford(2016)]%
        {crawford2016graphical}
\bibfield{author}{\bibinfo{person}{Forrest~W Crawford}.} \bibinfo{year}{2016}\natexlab{}.
\newblock \showarticletitle{The graphical structure of respondent-driven sampling}.
\newblock \bibinfo{journal}{\emph{Sociological methodology}} \bibinfo{volume}{46}, \bibinfo{number}{1} (\bibinfo{year}{2016}), \bibinfo{pages}{187--211}.
\newblock


\bibitem[Erdős and Rényi(1959)]%
        {erdosrenyi}
\bibfield{author}{\bibinfo{person}{Paul Erdős} {and} \bibinfo{person}{Alfred Rényi}.} \bibinfo{year}{1959}\natexlab{}.
\newblock \showarticletitle{On Random Graphs I}.
\newblock \bibinfo{journal}{\emph{Publicationes Mathematicae}}  \bibinfo{volume}{6} (\bibinfo{year}{1959}), \bibinfo{pages}{290–297}.
\newblock


\bibitem[Ezoe et~al\mbox{.}(2012)]%
        {ezoe2012population}
\bibfield{author}{\bibinfo{person}{Satoshi Ezoe}, \bibinfo{person}{Takeo Morooka}, \bibinfo{person}{Tatsuya Noda}, \bibinfo{person}{Miriam~Lewis Sabin}, {and} \bibinfo{person}{Soichi Koike}.} \bibinfo{year}{2012}\natexlab{}.
\newblock \showarticletitle{Population size estimation of men who have sex with men through the network scale-up method in Japan}.
\newblock \bibinfo{journal}{\emph{PloS one}} \bibinfo{volume}{7}, \bibinfo{number}{1} (\bibinfo{year}{2012}), \bibinfo{pages}{e31184}.
\newblock


\bibitem[Garbe and Vondrak(2018)]%
        {garbe2018concentration}
\bibfield{author}{\bibinfo{person}{Kevin Garbe} {and} \bibinfo{person}{Jan Vondrak}.} \bibinfo{year}{2018}\natexlab{}.
\newblock \showarticletitle{Concentration of Lipschitz functions of negatively dependent variables}.
\newblock \bibinfo{journal}{\emph{arXiv preprint arXiv:1804.10084}} (\bibinfo{year}{2018}).
\newblock


\bibitem[Garcia-Agundez et~al\mbox{.}(2021)]%
        {garcia2021estimating}
\bibfield{author}{\bibinfo{person}{Augusto Garcia-Agundez}, \bibinfo{person}{Oluwasegun Ojo}, \bibinfo{person}{Harold~A Hern{\'a}ndez-Roig}, \bibinfo{person}{Carlos Baquero}, \bibinfo{person}{Davide Frey}, \bibinfo{person}{Chryssis Georgiou}, \bibinfo{person}{Mathieu Goessens}, \bibinfo{person}{Rosa~E Lillo}, \bibinfo{person}{Raquel Menezes}, \bibinfo{person}{Nicolas Nicolaou}, {et~al\mbox{.}}} \bibinfo{year}{2021}\natexlab{}.
\newblock \showarticletitle{Estimating the COVID-19 prevalence in Spain with indirect reporting via open surveys}.
\newblock \bibinfo{journal}{\emph{Frontiers in Public Health}}  \bibinfo{volume}{9} (\bibinfo{year}{2021}), \bibinfo{pages}{658544}.
\newblock


\bibitem[Hoeffding(1994)]%
        {hoeffding1994probability}
\bibfield{author}{\bibinfo{person}{Wassily Hoeffding}.} \bibinfo{year}{1994}\natexlab{}.
\newblock \bibinfo{booktitle}{\emph{Probability Inequalities for sums of Bounded Random Variables}}.
\newblock \bibinfo{publisher}{Springer New York}, \bibinfo{address}{New York, NY}, \bibinfo{pages}{409--426}.
\newblock
\showISBNx{978-1-4612-0865-5}
\href{https://doi.org/10.1007/978-1-4612-0865-5_26}{doi:\nolinkurl{10.1007/978-1-4612-0865-5_26}}


\bibitem[Killworth et~al\mbox{.}(1998a)]%
        {killworth1998social}
\bibfield{author}{\bibinfo{person}{Peter~D Killworth}, \bibinfo{person}{Eugene~C Johnsen}, \bibinfo{person}{Christopher McCarty}, \bibinfo{person}{Gene~Ann Shelley}, {and} \bibinfo{person}{H~Russell Bernard}.} \bibinfo{year}{1998}\natexlab{a}.
\newblock \showarticletitle{A social network approach to estimating seroprevalence in the United States}.
\newblock \bibinfo{journal}{\emph{Social networks}} \bibinfo{volume}{20}, \bibinfo{number}{1} (\bibinfo{year}{1998}), \bibinfo{pages}{23--50}.
\newblock


\bibitem[Killworth et~al\mbox{.}(1998b)]%
        {killworth1998estimation}
\bibfield{author}{\bibinfo{person}{Peter~D Killworth}, \bibinfo{person}{Christopher McCarty}, \bibinfo{person}{H~Russell Bernard}, \bibinfo{person}{Gene~Ann Shelley}, {and} \bibinfo{person}{Eugene~C Johnsen}.} \bibinfo{year}{1998}\natexlab{b}.
\newblock \showarticletitle{Estimation of seroprevalence, rape, and homelessness in the United States using a social network approach}.
\newblock \bibinfo{journal}{\emph{Evaluation review}} \bibinfo{volume}{22}, \bibinfo{number}{2} (\bibinfo{year}{1998}), \bibinfo{pages}{289--308}.
\newblock


\bibitem[Laga et~al\mbox{.}(2021)]%
        {laga2021thirty-2021-jasa}
\bibfield{author}{\bibinfo{person}{Ian Laga}, \bibinfo{person}{Le Bao}, {and} \bibinfo{person}{Xiaoyue Niu}.} \bibinfo{year}{2021}\natexlab{}.
\newblock \showarticletitle{Thirty years of the network scale-up method}.
\newblock \bibinfo{journal}{\emph{J. Amer. Statist. Assoc.}} \bibinfo{volume}{116}, \bibinfo{number}{535} (\bibinfo{year}{2021}), \bibinfo{pages}{1548--1559}.
\newblock


\bibitem[Maltiel et~al\mbox{.}(2015)]%
        {maltiel2015estimating}
\bibfield{author}{\bibinfo{person}{Rachael Maltiel}, \bibinfo{person}{Adrian~E Raftery}, \bibinfo{person}{Tyler~H McCormick}, {and} \bibinfo{person}{Aaron~J Baraff}.} \bibinfo{year}{2015}\natexlab{}.
\newblock \showarticletitle{Estimating population size using the network scale up method}.
\newblock \bibinfo{journal}{\emph{The Annals of Applied Statistics}} \bibinfo{volume}{9}, \bibinfo{number}{3} (\bibinfo{year}{2015}), \bibinfo{pages}{1247}.
\newblock


\bibitem[McCormick et~al\mbox{.}(2010)]%
        {mccormick2010many}
\bibfield{author}{\bibinfo{person}{Tyler~H McCormick}, \bibinfo{person}{Matthew~J Salganik}, {and} \bibinfo{person}{Tian Zheng}.} \bibinfo{year}{2010}\natexlab{}.
\newblock \showarticletitle{How many people do you know?: Efficiently estimating personal network size}.
\newblock \bibinfo{journal}{\emph{J. Amer. Statist. Assoc.}} \bibinfo{volume}{105}, \bibinfo{number}{489} (\bibinfo{year}{2010}), \bibinfo{pages}{59--70}.
\newblock


\bibitem[McCormick and Zheng(2012)]%
        {mccormick2012latent}
\bibfield{author}{\bibinfo{person}{Tyler~H McCormick} {and} \bibinfo{person}{Tian Zheng}.} \bibinfo{year}{2012}\natexlab{}.
\newblock \showarticletitle{Latent demographic profile estimation in hard-to-reach groups}.
\newblock \bibinfo{journal}{\emph{The annals of applied statistics}} \bibinfo{volume}{6}, \bibinfo{number}{4} (\bibinfo{year}{2012}), \bibinfo{pages}{1795}.
\newblock


\bibitem[Panconesi and Srinivasan(1997)]%
        {Panconesi1997}
\bibfield{author}{\bibinfo{person}{Alessandro Panconesi} {and} \bibinfo{person}{Aravind Srinivasan}.} \bibinfo{year}{1997}\natexlab{}.
\newblock \showarticletitle{Randomized Distributed Edge Coloring via an Extension of the Chernoff--Hoeffding Bounds}.
\newblock \bibinfo{journal}{\emph{SIAM J. Comput.}} \bibinfo{volume}{26}, \bibinfo{number}{2} (\bibinfo{year}{1997}), \bibinfo{pages}{350--368}.
\newblock


\bibitem[Rozemberczki et~al\mbox{.}(2020)]%
        {rozemberczki2019gemsec}
\bibfield{author}{\bibinfo{person}{Benedek Rozemberczki}, \bibinfo{person}{Ryan Davies}, \bibinfo{person}{Rik Sarkar}, {and} \bibinfo{person}{Charles Sutton}.} \bibinfo{year}{2020}\natexlab{}.
\newblock \showarticletitle{GEMSEC: graph embedding with self clustering}. In \bibinfo{booktitle}{\emph{Proceedings of the 2019 IEEE/ACM International Conference on Advances in Social Networks Analysis and Mining}} (Vancouver, British Columbia, Canada) \emph{(\bibinfo{series}{ASONAM '19})}. \bibinfo{publisher}{Association for Computing Machinery}, \bibinfo{address}{New York, NY, USA}, \bibinfo{pages}{65–72}.
\newblock
\showISBNx{9781450368681}
\href{https://doi.org/10.1145/3341161.3342890}{doi:\nolinkurl{10.1145/3341161.3342890}}


\bibitem[Snidero et~al\mbox{.}(2012)]%
        {snidero2012scale}
\bibfield{author}{\bibinfo{person}{Silvia Snidero}, \bibinfo{person}{Nicola Soriani}, \bibinfo{person}{Ileana Baldi}, \bibinfo{person}{Federica Zobec}, \bibinfo{person}{Paola Berchialla}, {and} \bibinfo{person}{Dario Gregori}.} \bibinfo{year}{2012}\natexlab{}.
\newblock \showarticletitle{Scale-up approach in CATI surveys for estimating the number of foreign body injuries in the aero-digestive tract in children}.
\newblock \bibinfo{journal}{\emph{International journal of environmental research and public health}} \bibinfo{volume}{9}, \bibinfo{number}{11} (\bibinfo{year}{2012}), \bibinfo{pages}{4056--4067}.
\newblock


\bibitem[Srivastava et~al\mbox{.}(2024)]%
        {srivastava2024nowcasting}
\bibfield{author}{\bibinfo{person}{Ajitesh Srivastava}, \bibinfo{person}{Juan~Marcos Ramirez}, \bibinfo{person}{Sergio D{\'\i}az-Aranda}, \bibinfo{person}{Jose Aguilar}, \bibinfo{person}{Antonio~Fern{\'a}ndez Anta}, \bibinfo{person}{Antonio Ortega}, {and} \bibinfo{person}{Rosa~Elvira Lillo}.} \bibinfo{year}{2024}\natexlab{}.
\newblock \showarticletitle{Nowcasting temporal trends using indirect surveys}. In \bibinfo{booktitle}{\emph{Proceedings of the AAAI Conference on Artificial Intelligence}}, Vol.~\bibinfo{volume}{38}. \bibinfo{pages}{22359--22367}.
\newblock


\bibitem[UNAIDS(2010)]%
        {unaids2010guidelines}
\bibfield{author}{\bibinfo{person}{WHO UNAIDS}.} \bibinfo{year}{2010}\natexlab{}.
\newblock \bibinfo{title}{Guidelines on Estimating the Size of Populations most at Risk to {HIV}. {Geneva: World Health Organization}}.
\newblock


\bibitem[Wang et~al\mbox{.}(2015)]%
        {wang2015application}
\bibfield{author}{\bibinfo{person}{Jun Wang}, \bibinfo{person}{Ying Yang}, \bibinfo{person}{Wan Zhao}, \bibinfo{person}{Hualin Su}, \bibinfo{person}{Yanping Zhao}, \bibinfo{person}{Yue Chen}, \bibinfo{person}{Tao Zhang}, {and} \bibinfo{person}{Tiejun Zhang}.} \bibinfo{year}{2015}\natexlab{}.
\newblock \showarticletitle{Application of network scale up method in the estimation of population size for men who have sex with men in Shanghai, China}.
\newblock \bibinfo{journal}{\emph{PloS one}} \bibinfo{volume}{10}, \bibinfo{number}{11} (\bibinfo{year}{2015}), \bibinfo{pages}{e0143118}.
\newblock


\bibitem[Zheng et~al\mbox{.}(2006)]%
        {zheng2006many}
\bibfield{author}{\bibinfo{person}{Tian Zheng}, \bibinfo{person}{Matthew~J Salganik}, {and} \bibinfo{person}{Andrew Gelman}.} \bibinfo{year}{2006}\natexlab{}.
\newblock \showarticletitle{How many people do you know in prison? Using overdispersion in count data to estimate social structure in networks}.
\newblock \bibinfo{journal}{\emph{J. Amer. Statist. Assoc.}} \bibinfo{volume}{101}, \bibinfo{number}{474} (\bibinfo{year}{2006}), \bibinfo{pages}{409--423}.
\newblock


\end{thebibliography}
\bibliographystyle{ACM-Reference-Format}
\balance
\newpage

\appendix
\section*{Appendices}

\allowdisplaybreaks
\section{Concentration Bounds for Negatively Correlated Random Variables}

Let us define the negative correlation of random variables (rvs).

\begin{definition}
\label{def:negcorrel}
The rvs $Z_1,Z_2,\ldots,Z_n$ are \textit{negatively correlated} if for any subset $B \subseteq \{1,2,\ldots,n\}$, it holds that $$\E\Brack{\prod_{i \in B} Z_i} \leq \prod_{i \in B} \E\Brack{Z_i}.$$
\end{definition}

\begin{definition}
    \label{def:negcyl}
    The rvs $Z_1,Z_2,\ldots,Z_n$ are said to exhibit \textit{negative cylinder dependence} \cite{garbe2018concentration} if both sets of rvs $Z_1,Z_2,\ldots,Z_n$ and $1-Z_1,1-Z_2,\ldots,1-Z_n$ are negatively correlated.
\end{definition}

We use the multiplicative error Chernoff bound for bounded rvs $0 \leq Z_i \leq 1$ with negative cylinder dependence presented in the following theorem. The proof of the theorem is a modification of the result for $0$-$1$ rvs from \cite{Panconesi1997} and is presented below for completeness.

\begin{theorem}\label{eq:chernoff-twosided}
Let $Z_1,Z_2,\ldots,Z_n$ be rvs with negative cylinder dependence, where $0\le Z_i \le 1$. Let $Z = \sum_{i=1}^n Z_i$ and let $\mu = \E[Z]$, then for $\beta > 1$, it holds that
$$\Pr\Brack{\neg\paren{\frac{\mu}{\beta}\leq Z \leq \beta\mu}} \leq \Paren{\frac{e^{\beta-1}}{\beta^\beta}}^{\mu} + \Paren{\frac{e^{\frac{1}{\beta}-1}}{\beta^{-1/\beta}}}^{\mu} = F(\beta, \mu),$$
where $\Pr[\neg A]$ is the probability of $A$ complement.
\label{thm:multiplicative-error-conc-bound}
\end{theorem}

\begin{proof}
Using Markov's inequality, for $t \ge 0$, we have
\begin{align}\label{eq1:thmB1}
\Pr[Z\geq \beta \mu] = \Pr\Brack{e^{tZ} \geq e^{t\beta\mu}} \leq \frac{\E\Brack{e^{tZ}}}{e^{t\beta\mu}},
\end{align}
$$\mathrm{where~}\E\Brack{e^{tZ}} = \E\Brack{e^{t\sum_{i=1}^nZ_i}} = \E\Brack{\prod_{i=1}^{n}e^{tZ_i}}.$$
We prove that $\E\Brack{\prod_{i=1}^{n}e^{tZ_i}} \leq  \prod_{i=1}^{n}\E\Brack{e^{tZ_i}}$.
Let ${\hat{Z}_1,\ldots,\hat{Z}_n}$ be independent $0$-$1$ rvs such that $\E\Brack{\hat{Z}_i} = \E\Brack{Z_i}$ for each $i = \curly{1,2\ldots,n}$, and define $\hat{Z}:= \hat{Z}_1 + \cdots + \hat{Z}_n$. 
For $k > 0$, $\mathbf{\alpha} = \{\alpha_1,\alpha_2,\ldots,\alpha_n\}$ is a decomposition such that $\sum_{i=1}^n \alpha_i = k$, where $\alpha_i \geq 0$. We have $\E\Brack{Z^k} = \sum_{\boldsymbol{\alpha}}\E\Brack{\prod_{i=1}^{n} Z_i^{\alpha_i}}.$ We claim that for any $\boldsymbol{\alpha}$, we have $\E\Brack{\prod_{i=1}^{n} Z_i^{\alpha_i}} \leq \E\Brack{ \prod_{i=1}^{n} \hat{Z}_i^{\alpha_i}}.$  
Without loss of generality, for any $\boldsymbol{\alpha}$, we have $\alpha_i > 0$, for all $i \le l$, and $\alpha_i = 0$ for $i > l$, 
for some $l \le n$. Therefore, we have
\begin{multline}
    \E\Brack{\prod_{i=1}^{n} Z_i^{\alpha_i}} = \E\Brack{\prod_{i=1}^{l} Z_i^{\alpha_i}} 
    \le \E\Brack{\prod_{i=1}^{l} Z_i} \le \prod_{i=1}^{l} \E\Brack{Z_i}
     = \prod_{i=1}^{l} \E\Brack{\hat{Z}_i} = \E\Brack{ \prod_{i=1}^{l} \hat{Z}_i^{\alpha_i}}
    = \E\Brack{ \prod_{i=1}^{n} \hat{Z}_i^{\alpha_i}}. \nonumber
\end{multline}
We used $Z_i \in [0,1]$ for the second step above, and for the third step, we used the negative correlation property. For the fifth step, we used the fact that $\hat{Z}_i$ are $0$-$1$ independent rvs. We have 
\begin{eqnarray*}
    \E\Brack{Z^k} \leq \sum_{\boldsymbol{\alpha}}\E\Brack{\prod_{i=1}^{n} Z_i^{\alpha_i}} \leq \sum_{\boldsymbol{\alpha}}\E\Brack{\prod_{i=1}^{n} \hat{Z}_i^{\alpha_i}} = \E\Brack{\hat{Z}^k}.
\end{eqnarray*}
Thus, for any $k \geq 1$, we have $\E\Brack{Z^k} \le \E\Brack{\hat{Z}^k}$. Using this result in the Taylor expansion of $e^{t\hat{Z}}$ in $\E\brac{e^{t\hat{Z}}}$, we obtain
$\E\Brack{e^{tZ}}\leq \E\Brack{e^{t\hat{Z}}}$. We now bound $\E\Brack{e^{t\hat{Z}_i}}$. Since $\hat{Z}_i$ are 0-1 rvs, we have $\Pr(\hat{Z}_i = 1) =  \E\Brack{\hat{Z}_i}$ and $\Pr(\hat{Z}_i = 0) = 1- \E\Brack{\hat{Z}_i}$. We have
\begin{align*}
\E\Brack{e^{t\hat{Z}_i}} &= e^t\E\Brack{\hat{Z}_i} + \E\Brack{\hat{Z}_i}=\E\Brack{Z_i}(e^{t}-1) + 1 \leq e^{\E\Brack{Z_i}(e^{t}-1)}.
\end{align*}
We used $1+x \leq e^x$ for all $x\in \mathbb{R}$ in the last inequality.
Hence, 
\begin{align}\label{eq2:thmB1}
\E\brac{e^{tZ}} \leq \prod_{i=1}^{n}e^{\E\Brack{Z_i}(e^t-1)} = e^{\mu(e^t-1)}.    
\end{align}
Substituting \eqref{eq2:thmB1} in \eqref{eq1:thmB1} and choosing $t = \ln \beta$, we obtain
\begin{align}\label{eq3:thmB1}
\Pr\Brack{Z\geq \beta \mu} \leq \Paren{\frac{e^{\beta-1}}{\beta^\beta}}^{\mu}.    
\end{align}

Again, from Markov's inequality, we have $$\Pr\brac{Z\le \frac{\mu}{\beta}} \leq e^{\frac{t\mu}{\beta}}\E\brac{e^{-tZ}}.$$ 
Given that variables $1-Z_i$ are negatively correlated due to the negative cylinder dependence, we can apply the steps outlined in the previous case to obtain that
$$\E\Brack{e^{t\sum_{i=1}^n(1-Z_i)}} \leq \E\Brack{e^{t\sum_{i=1}^n(1-\hat{Z}_i)}},$$
and then
\begin{align*}
\E\Brack{e^{-tZ}} = e^{-tn} \E\Brack{e^{tn-tZ}} =e^{-tn} \E\Brack{e^{t\sum_{i=1}^n(1-Z_i)}}  \leq e^{-tn} \E\Brack{e^{t\sum_{i=1}^n(1-\hat{Z}_i)}} = e^{-tn} \E\Brack{e^{tn-t\hat{Z}}} = \E\Brack{e^{-t\hat{Z}}}.    
\end{align*}
 
Following similar steps as above, we obtain
\begin{align}\label{eq4:thmB1}
    \Pr\brac{Z\le \frac{\mu}{\beta}} \leq \Paren{\frac{e^{\frac{1}{\beta}-1}}{\beta^{-1/\beta}}}^{\mu}.
\end{align}
From \eqref{eq3:thmB1} and \eqref{eq4:thmB1} and the fact that $\Pr(\neg(A \cap B)) \leq \Pr(\neg A) + \Pr(\neg B)$.
\end{proof}

The following corollary uses $\beta = 1/(1-\delta)$ in \eqref{eq4:thmB1}, for  $\delta \in (0,1)$.
\begin{corollary}\label{eq:chernoff-onesided}
Let $Z_1,Z_2,\ldots,Z_n$ be rvs with negative cylinder dependence, where $0\le Z_i \le 1$. Let $Z = \sum_{i=1}^n Z_i$ and let $\mu = \E[Z]$, then for $\delta \in (0,1)$, it holds that
$$\Pr\Brack{Z \leq (1-\delta) \mu} \leq \Paren{\frac{e^{-\delta}}{(1-\delta)^{(1-\delta)}}}^{\mu}.$$
\label{thm:lower-conc-bound}
\end{corollary}




\allowdisplaybreaks

\renewcommand{\proofinappendix}{}

\section{Proofs of Lemma \ref{lem:exp-Y-v} and Theorem \ref{thm:MoR-bound-main}}

\begin{lemma}
\label{lem:exp-X-dep}
If $v \in H$ then $\E[X_{vj} \mid v \in H] = \frac{h-1}{n-1}$. If $v \notin H$ then $\E[X_{vj} \mid v \notin H] = \frac{h}{n-1}$.
\end{lemma}

\begin{proof}
Let us assume $v \in H$. 
By definition of the expectation of an indicator variable, $\E[X_{vj} \mid v \in H]$ is the probability that the $j$th in-neighbor selected by $v$ belongs to $H$. This is 
\begin{equation}
    \E[X_{vj} \mid v \in H] = \frac{(h-1)(n-2)!}{(n-1)!} = \frac{h-1}{n-1},
\end{equation}
where the numerator is the possible permutations of the $n-1$ vertices in $V\setminus \{v\}$ in which the $j$th vertex is in $H$, and the denominator is the number of possible permutations of the $n-1$ vertices in $V\setminus \{v\}$.
The same argument can be applied to the case $v \notin H$ to obtain $\E[X_{vj} \mid v \notin H] = \frac{h}{n-1}$.
\end{proof}



\lemexpXYv*

\begin{proof}
\begin{align*}
\E[X_{vj}] &= \Pr[v \in H] \cdot \E[X_{vj}  \mid v \in H] +
\Pr[v \notin H] \cdot \E[X_{vj}  \mid v \notin H] \\
&= \frac{h}{n} \frac{h-1}{n-1} + \frac{n-h}{n} \frac{h}{n-1} = \frac{h}{n} \left( \frac{h-1}{n-1} + \frac{n-h}{n-1} \right) = \frac{h}{n} = \ratio{(I)}
\end{align*}
\begin{align*}
\E[Y_v] &= \sum_r \sum_c \frac{c}{r} \Pr[\cases{v}=c,\reach{v}=r] \\
&= \sum_r \frac{1}{r} \Pr[\reach{v}=r] \cdot \sum_c c \Pr[\cases{v}=c \mid \reach{v}=r] \\
&= \sum_r \frac{1}{r} \Pr[\reach{v}=r] \cdot \E[\cases{v} \mid \reach{v}=r] \\
&= \sum_r \frac{1}{r} \Pr[\reach{v}=r] \cdot \sum_{j=1}^r \E[X_{vj}] \\
&= \sum_r \frac{1}{r} \Pr[\reach{v}=r] r \ratio{(I)} = \sum_r \Pr[\reach{v}=r] \ratio{(I)} = \ratio{(I)}.
\end{align*}
\end{proof}

\begin{lemma}\label{lem:Y-negativeCorr}
Let $S \subseteq V$ be a set with size $m$ chosen uniformly at random, then the rvs $Y_v$, for all $v \in S$, are negatively correlated, i.e., $\forall S_0 \subseteq S, \E[\prod_{v \in S_0} Y_v] \leq
\prod_{v \in S_0} \E[Y_v]$.
\end{lemma}
\begin{proof}
Let $S_0 \subseteq S$ be a set of nodes. The variables $Y_v$, for $v \in S_0$, only depend on whether $v$ belongs to $H$, and they are mutually independent. Then, given a particular configuration $A=(S_0 \cap H, S_0 \cap \Bar{H})$ of hidden-population node memberships, the expectation of the product of the variables $Y_v$, conditional on this configuration, is equal to
\begin{eqnarray*}
\E \left[ \prod_{v \in S_0} Y_v \mid A \right] &=&
        \prod_{v\in S_0 \cap H}E[Y_v|A] \prod_{v \in S_0 \cap \Bar{H}} E[Y_v|A] \\
         &=& 
        \prod_{v\in S_0 \cap H}E[Y_v|v\in H] \prod_{v \in S_0 \cap \Bar{H}} E[Y_v|v\not\in H] \\
        &=& \left(\frac{h-1}{n-1}\right)^{|S_0 \cap H |} \left(\frac{h}{n-1}\right)^{|S_0 \cap \Bar{H} |}. \nonumber
\end{eqnarray*}

Let $Z = |S_0 \cap H|$ and $m_0 = |S_0|$. Since we do not know if the nodes are in H, $Z$ is a hypergeometric random variable with probability mass function $\Pr[Z=z] = \frac{\binom{h}{z} \binom{n-h}{m_0-z}}{\binom{n}{m_0}}$  and $\E[Z] = \frac{m_0h}{n}$. Let $\lbrace{ A_i \rbrace}$ be the partition of the sample space composed of all the individual cases in which $Z=z$. Note that $\sum_i \Pr[A_i]= \Pr[Z=z]$. We can apply the Law of total expectation to $\prod_{v \in S_0} Y_v\mid Z=z$ over the sample space restricted $Z=z$. Then, 
\begin{align*}
    \E \left[ \prod_{v \in S_0} Y_v \mid Z=z \right] & = \sum_{i} \E \left[ \prod_{v \in S_0} Y_v \mid A_i, Z=z\right] \Pr[A_i\mid Z=z]\\
    & = \sum_{i} \E \left[ \prod_{v \in S_0} Y_v \mid A_i\right] \Pr[A_i\mid Z=z]\\
    & = \sum_i \left(\frac{h-1}{n-1}\right)^z \left(\frac{h}{n-1}\right)^{m_0-z} \Pr[A_i\mid Z=z]\\
    & = \sum_i \left(\frac{h-1}{n-1}\right)^z \left(\frac{h}{n-1}\right)^{m_0-z} \frac{\Pr[A_i]}{\Pr[Z=z]}\\
    & = \left(\frac{h-1}{n-1}\right)^z \left(\frac{h}{n-1}\right)^{m_0-z}.
\end{align*}
The unconditional expectation of the product is
\begin{align}\label{eq1:negcorr}
    \E\left[\prod_{v \in S_0} Y_v \right]
    &= \sum_{z=0}^{m_0} \E\left[\prod_{v \in S_0} Y_v \mid Z=z \right]\Pr[Z=z] \nonumber\\
    &= \sum_{z=0}^{m_0} \left(\frac{h-1}{n-1}\right)^z \left(\frac{h}{n-1}\right)^{m_0-z}\Pr[Z=z] \nonumber\\
    &= \frac{h^{m_0}}{(n-1)^{m_0}} \sum_{z=0}^{m_0} (1-h^{-1})^z \Pr[Z=z]  \nonumber\\
    &= \frac{h^{m_0}}{(n-1)^{m_0}} \E[e^{\ln (1-h^{-1})Z}].
\end{align}
$\E[e^{tZ}]$ is the moment generating function (MGF) of $Z$. We relate the MGF of $Z$ with the MGF of a binomial random variable $X$ with parameters $(m_0,\frac{h}{n})$. Note that $\E[X] = \E[Z]$. We state the claim.

    \begin{claim}
        MGF of $X$ is at least the MGF of $Z$, i.e., $\E[e^{tX}] \geq \E[e^{tZ}]$, for all $t \in \mathbb{R}$.
    \end{claim}
    \begin{proof}
        By Hoeffding \cite{hoeffding1994probability}, $\E[f(\sum Z_i)] \geq \E[f(\sum X_i)]$, where $Z_i$ are random samples without replacement from a finite population $C$, $X_i$ are random samples with replacement from $C$, and $f$ is a convex function. The claim follows from the fact that $X = \sum_{i=1}^m X_i$, with $X_i \sim Bernoulli(\frac{h}{n})$, $Z=\sum_{i=1}^m Z_i$, where $Z_i$ are variables without replacement indicating the membership in $H$ of the ith draw, and $x \mapsto exp(tx)$ is a convex function. 
    \end{proof}

The MGF of the binomial random variable $X$ is given by 
$$\E[e^{tX}] = \left(1-\frac{h}{n}+\frac{he^t}{n}\right)^m.$$ 
From the claim, we have
\begin{align}\label{eq2:negcorr}
    \E[e^{(1-h^{-1})Z}] &\leq \E[e^{\ln (1-h^{-1})X}] \nonumber = \left(1-\frac{h}{n}+\frac{he^{\ln (1-h^{-1})}}{n}\right)^m_0 \nonumber\\
    &= \left(1-\frac{h}{n}+\frac{h}{n}(1-h^{-1})\right)^m_0 \nonumber = \frac{(n-1)^{m_0}}{n^{m_0}}.
\end{align}
Using this in \eqref{eq1:negcorr}, we obtain
\begin{align*}
        \E\left[\prod_{v \in S_0} Y_v \right] \leq & \frac{h^{m_0}}{(n-1)^{m_0}} \frac{(n-1)^{m_0}}{n^{m_0}} = \left(\frac{h}{n}\right)^{m_0} = \prod_{v \in S_0} \E[Y_v].
\end{align*} 
\end{proof}

\begin{lemma}\label{lem:Y-negcyl}
  Let $S \subseteq V$ be a set with size $m$ chosen uniformly at random, then the rvs $Y_v$, for all $v \in S$ exhibit negative cylinder dependence, i.e., $\forall S_0 \subseteq S$
\begin{enumerate}
    \item[\text{(i)}] $\E[\prod_{v \in S_0} Y_v] \leq \prod_{v \in S_0} \E[Y_v]$ 
    \item[\text{(ii)}] $\E[\prod_{v \in S_0} (1-Y_v)] \leq 
\prod_{v \in S_0} \E[1-Y_v]$
\end{enumerate}
\end{lemma}
\begin{proof}
    The first part is proven in Lemma \ref{lem:Y-negativeCorr}. For the negative correlation of the $1-Y_v$, we observe that the variables $1-Y_v$ correspond to the variables $Y_v'$ from a random network obtained by changing the hidden population to its complement. That is, an instance $(G',H')$, with degree distribution $P'_\text{deg}$, and variables $R'_v$ and $C'_v$, where $G'=G$, $H'=\Bar{H}$, $P'_\text{deg}= P_\text{deg}$, $R'_v = R_v$ and $C'_v = R_v-C_v$. Thus, applying Lemma \ref{lem:Y-negativeCorr} to the variables of this random network, we conclude the proof.
\end{proof}

\thmMoRboundmain*

\begin{proof}
Given $\E[Y_v] = \rho(I)$ and Lemma \ref{lem:Y-negcyl}, we apply \Cref{thm:multiplicative-error-conc-bound} to the rv 
$Y_v$, and $Y_S=\sum_{v \in S} Y_v$ with $\mu=\E[Y_S]=m \ratio{(I)}$.
\begin{align*}
    \Pr\Brack{\error{\MoR}{(I)}{} > \beta} &= \Pr\Brack{\max\left(\frac{\ratio{(I)}}{\estimate{\MoR}{(I)}},\frac{\estimate{\MoR}{(I)}}{\ratio{(I)}}\right) > \beta}\\
    &= \Pr\Brack{\neg\paren{\frac{\ratio{(I)}}{\beta}\leq \estimate{\MoR}{(I)} \leq \ratio{(I)}\beta}} \\
    &= \Pr\Brack{\neg\paren{\frac{m \ratio{(I)}}{\beta}\leq Y_S \leq m \ratio{(I)}\beta}} \\
    &\leq \left(\frac{e^{\beta-1}}{\beta^{\beta}}\right)^{m \ratio{(I)}} + \left(\frac{e^{\frac{1}{\beta}-1}}{\beta^{-1/\beta}}\right)^{m \ratio{(I)}} \\
    & = F(\beta, m \ratio{(I)}).
\end{align*}
\end{proof}

\section{Proof of Theorem
\ref{thm:RoS-bound-main}}

\begin{lemma}\label{lem:X-negativeCorr}
Let $S \subseteq V$ be a set with size $m$ chosen uniformly at random, then for all $v \in S$ and $j \in \{1,2,\ldots,R_v\}$, the rvs $X_{vj}$ are negatively correlated, i.e., for all $S_0 \subseteq S$ and $B_v\subseteq \{1,\cdots,R_v\}$ with $v\in S_0$, 
$\E\Brack{\prod_{v \in S_0} \prod_{j=\in B_v} X_{vj}} \leq \prod_{v \in S_0} \prod_{B_v} \E[X_{vj}]$.
\end{lemma}
\begin{proof}
Let $S_0 \subseteq S$ be a set of nodes, and $m_0$ the size of $S_0$. For notational simplicity, we consider only the case in which each index set takes the form $B_v =\{ 1,\cdots, R_v\}$. Observe that the variables for different nodes $X_{vj}$ may be dependent, but are independent if the membership in $H$ is known. The conditional expectation of the product $\prod_{v \in S_0} \prod_{j=1}^{\reach{v}} X_{vj}$ is equal to
\begin{equation}
    \left( \prod_{v \in S_0 \cap {H}} \E \left[ \prod_{j=1}^{\reach{v}} X_{vj} \mid v\in H \right] \right) \left( \prod_{v \in S_0 \cap \Bar{H}} \E \left[ \prod_{j=1}^{\reach{v}} X_{vj} \mid v\not\in H \right] \right) .
\end{equation}
We also have that
\begin{equation*}
    \E \left[ \prod_{j=1}^{\reach{v}} X_{vj} \mid v\in H \right] = \Pr \left[ \prod_{j=1}^{\reach{v}} X_{vj} = 1 \mid v\in H \right] =  \prod_{j=1}^{\reach{v}} \frac{h-j}{n-j} \leq \ratio{(I)}^{\Reach{v}}, \nonumber
\end{equation*}
and
\begin{equation*}
    \E \left[ \prod_{j=1}^{\reach{v}} X_{vj} \mid v\not\in H \right] = \Pr \left[ \prod_{j=1}^{\reach{v}} X_{vj} = 1 \mid v\not\in H \right] =  \prod_{j=1}^{\reach{v}} \frac{h-j-1}{n-j} \leq \ratio{(I)}^{\Reach{v}}. \nonumber
\end{equation*}
Let $\lbrace{B_i \rbrace}$ be a partition of the sample space composed of all the possible combinations of nodes in $H$. Then,
\begin{align*}
     \E\left[ \prod_{v \in S_0} \prod_{j=1}^{\reach{v}} X_{vj} \right] & = \sum_{i}  \E\left[ \prod_{v \in S_0} \prod_{j=1}^{\reach{v}} X_{vj} \mid B_i \right] \Pr[B_i]\\
     & = \sum_{i} \prod_{v \in S_0} \E\left[ \prod_{j=1}^{\reach{v}} X_{vj} \mid B_i \right] \Pr[B_i]\\
     & \leq  \sum_{i} \prod_{v \in S_0} \ratio{(I)}^{R_v} \Pr[B_i]\\
     & = \prod_{v \in S_0} \ratio{(I)}^{R_v}  \\
     & = \ratio{(I)}^{\sum_{v\in S_0}R_v} \\
     & = \prod_{v \in S_0} \prod_{j=1}^{\reach{v}} \E[X_{vj}].
\end{align*}
\end{proof}

\begin{lemma}\label{lem:X-negcyl}
Let $S \subseteq V$ be a set with size $m$ chosen uniformly at random, then for all $v \in S$ and $j \in \{1,2,\ldots,R_v\}$, the random variables $X_{vj}$ have negative cylinder dependence, i.e., for all $S_0 \subseteq S,$ and $B_v\subseteq \{1,\cdots,R_v\}$ with $v\in S_0$
\begin{enumerate}
    \item[\text{(i)}] $\E\Brack{\prod_{v \in S_0} \prod_{j\in B_v} X_{vj}}  \leq \prod_{v \in S_0} \prod_{j\in B_v} \E[X_{vj}]$ 
    \item[\text{(ii)}] $\E\Brack{\prod_{v \in S_0} \prod_{j\in B_v} (1-X_{vj})}  \leq \prod_{v \in S_0} \prod_{j \in B_v} \E[1-X_{vj}]$ 
\end{enumerate}

\end{lemma}
\begin{proof}
    (i) is a consequence of Lemma \ref{lem:X-negativeCorr}. To prove (ii), we see that the variables $X'_{vj}= 1-X_{vj}$ are the variables of a random network constructed from an instance $(G,\Bar{H})$, degree distribution $P_\text{deg}$, and variables $R'_v =Rv$ and $C'_v = R_v-C_v$. 
\end{proof}

\thmRoSboundmain*

\begin{proof}
By definition, $\estimate{\RoS}{(I)}=\Cases{S} / \Reach{S}$ and $\Cases{S}= \sum_{v \in S} \cases{v}= \sum_{v \in S} \sum_{j=1}^{\reach{v}} X_{vj}$, and hence $\E[\Cases{S}]=\Reach{S} \ratio{(I)}$. Given Lemma \ref{lem:X-negcyl}, we apply \Cref{thm:multiplicative-error-conc-bound} to the random variables $\Cases{S}$.
\begin{align*}
    \Pr\Brack{\error{\RoS}{(I,S)}{} > \beta}  &=  \sum_{\Reach{}} \Pr\Brack{\error{\RoS}{(I)}{} > \beta \mid \Reach{S} = \Reach{}} \Pr\Brack{\Reach{S}=\Reach{}} \\ \nonumber
    & = \sum_{\Reach{}} \Pr\Brack{\neg\paren{\frac{\ratio{(I)}}{\beta}\leq \estimate{\RoS}{(I)} \leq \ratio{(I)}\beta} \mid \Reach{S}=\Reach{}} \Pr\Brack{\Reach{S}=\Reach{}}  \\
    & = \sum_{\Reach{}} \Pr\Brack{\neg\paren{\frac{\Reach{} \ratio{(I)}}{\beta}\leq \Cases{S} \leq \Reach{} \ratio{(I)} \beta} \mid \Reach{S}=\Reach{}} \Pr\Brack{\Reach{S}=\Reach{}} \\
    & \leq \sum_{\Reach{}} \left( \left(\frac{e^{\beta-1}}{\beta^{\beta}}\right)^{\Reach{} \ratio{(I)}} + \left(\frac{e^{\frac{1}{\beta}-1}}{\beta^{-1/\beta}}\right)^{\Reach{} \ratio{(I)}} \right) \Pr\Brack{\Reach{S}=\Reach{}}\\
    & = \sum_{\Reach{}} 
    F(\beta, \Reach{} \ratio{(I)})
    \Pr\Brack{\Reach{S}=\Reach{}}.
\end{align*}
\end{proof}

\end{document}